\documentclass[conference]{IEEEtran}
\IEEEoverridecommandlockouts
\usepackage{cite}
\usepackage{amsmath,amssymb,amsfonts}
\usepackage{graphicx}
\usepackage{textcomp}
\usepackage{xcolor}
\def\BibTeX{{\rm B\kern-.05em{\sc i\kern-.025em b}\kern-.08em
    T\kern-.1667em\lower.7ex\hbox{E}\kern-.125emX}}

\usepackage{amsmath,amssymb,amsfonts,amsthm}
\usepackage{algorithm, algorithmicx}
\usepackage[noend]{algpseudocode}

\usepackage{graphicx}
\usepackage{textcomp}
\usepackage{xcolor}
\usepackage{enumerate}
\usepackage{multirow}
\usepackage{booktabs}
\usepackage{epstopdf}
\usepackage{subfigure}
\usepackage{times}
\usepackage{enumitem}
\usepackage{makecell}

\usepackage[utf8]{inputenc}
\usepackage{pifont}

\usepackage{listings}
\usepackage{color}

\definecolor{dkgreen}{rgb}{0,0.6,0}
\definecolor{gray}{rgb}{0.5,0.5,0.5}
\definecolor{mauve}{rgb}{0.58,0,0.82}

\newtheorem{problem}{Problem}
\newtheorem{definition}{Definition}

\newtheorem{lemma}{Lemma}

\newtheorem{lemma*}{Lemma}
\newtheorem{corollary*}{Corollary}

\newcommand{\stitle}[1]{\vspace*{0.4em}\noindent{\bf #1.\/}}
\newcommand{\sstitle}[1]{\vspace*{0.4em}\noindent{\bf #1:\/}}

\newcommand{\trim}{\vspace{-2mm}}

\newcommand{\thor}{$\mathsf{CheetahGIS}$}
\newcommand{\thorm}{$\mathsf{CheetahGIS}^{-}$}
\newcommand{\statefun}{$\mathsf{StateFun}$}
\newcommand{\flink}{$\mathsf{Flink}$}
\newcommand{\bb}{{\sf QueryBroadcast}}

\newcommand{\oldversion}[1]{}

\newcommand{\squishlist}{
 \begin{list}{$\bullet$}
 { \setlength{\itemsep}{0pt}
   \setlength{\parsep}{3pt}
   \setlength{\topsep}{3pt}
   \setlength{\partopsep}{0pt}
   \setlength{\leftmargin}{1.2em}
   \setlength{\labelwidth}{1em}
   \setlength{\labelsep}{0.6em}
 }
}
\newcommand{\squishend}{
 \end{list}
}
\begin{document}

\title{CheetahGIS: Architecting a Scalable and Efficient Streaming Spatial Query Processing System}
\author{
    \IEEEauthorblockN{Jiaping Cao\IEEEauthorrefmark{1}, Ting Sun\IEEEauthorrefmark{2}, Man Lung Yiu\IEEEauthorrefmark{1}, Xiao Yan\IEEEauthorrefmark{3}, Bo Tang\IEEEauthorrefmark{2}}
    \IEEEauthorblockA{\IEEEauthorrefmark{1}Department of Computing, Hong Kong Polytechnic University
    \\\{csjcao1, csmlyiu\}@comp.polyu.edu.hk}
    \IEEEauthorblockA{\IEEEauthorrefmark{2}Department of Computer Science, Southern University of Science and Technology
    \\ suntcrick@gmail.com, tangb3@sustech.edu.cn}
    \IEEEauthorblockA{\IEEEauthorrefmark{3}Institute for Math \& AI, Wuhan University
    \\ yanxiaosunny@whu.edu.cn}
}

\maketitle
\thispagestyle{plain}
\pagestyle{plain}

\begin{abstract}
Spatial data analytics systems are widely studied in both the academia and industry.
However, existing systems are limited when handling a large number of moving objects and real-time spatial queries.
In this work, we architect a scalable and efficient system \thor{} to process streaming spatial queries over massive moving objects.
In particular, \thor{} is built upon Apache Flink Stateful Functions (\statefun{}), an API for building distributed streaming applications with an actor-like model.
\thor{} enjoys excellent scalability due to its modular architecture, which clearly decomposes different components and allows scaling individual components.
To improve the efficiency and scalability of \thor{}, we devise a suite of optimizations,
e.g., lightweight global grid-based index, metadata synchronization strategies, and load balance mechanisms.
We also formulate a generic paradigm for spatial query processing in \thor{}, and verify its generality by processing three representative streaming queries (i.e., object query, range count query, and $k$ nearest neighbor query).
We conduct extensive experiments on both real and synthetic datasets to evaluate \thor{}.
\end{abstract}

\section{Introduction}\label{sec:intro}

With the popularization of GPS devices, large-scale spatial data (e.g., for car, bus, and pedestrian) are collected and processed in the mobile Internet era.
To process spatial data efficiently, many algorithms~\cite{tao2013approximate,ouyang2020progressive} have been designed.
Meanwhile, many spatial data management systems (e.g., SpatialHadoop~\cite{eldawy2015spatialhadoop}, GeoSpark~\cite{yu2015geospark}, UlTraMan~\cite{ding2018ultraman}) have also been developed for the storage, management, and analysis of spatial data.
Recently, the demand for real-time query processing over a large number of moving objects is on the rise. 
We give two example applications as follows.

\stitle{Example I: Request a ride}
When a passenger requests a ride on the ride service apps (e.g., Uber or Didi), the service provider will issue an online $k$ nearest neighbor ($k$NN) query to find the top-$k$ nearest available taxis to the passenger. 
As the taxis are moving all the time, query results with large delay may cause location error and unpleasant user experiences.
To make matters worse, the scale of ride requests  is extremely large for ride services. 
For instance, Uber receives five million ride sharing requests per day worldwide~\cite{uberride}.
The service provider may also run other applications on streaming trip data, e.g., predicting trip patterns (e.g., trip density) ~\cite{uber2017}.

\stitle{Example II: Social networks}
Many social networks (e.g., Twitter, Wechat)  provide location-based services,
e.g., broadcasting real-time regional news and recommending Points-of-Interests.
The number of users in social networks can reach millions or even billions.
Each user can be modeled as a moving object, which dynamically updates its location. 
The technical challenges of location-based services in social networks are (i) how to maintain the real-time location of every user and (ii) how to cope with many location-based queries that are issued (by either user or system) constantly in time~\cite{thinkspatial}.

From two applications above, we summarize the key requirements (and the corresponding challenges) of real-time query processing over large-scale streaming data.

\squishlist
\item \sstitle{Efficiency} Short query latency is  crucial for the online services.
It is challenging as (i) the real-time spatial queries are conducted on a large number of moving objects with frequent location updates, and (ii) the computation of streaming spatial queries can be complex, e.g., counting the number of cars in a given region.

\item \sstitle{Scalability} Good scalability is essential as large-scale spatial data is usually beyond the capacity of a single machine.
Moreover, the query throughput should also scale with cluster resources.
It is challenging to build a scalable streaming spatial query processing system that does not suffer from severe performance degradation when processing tremendous data updates or spatial queries.

\item \sstitle{Generality} The generality of a streaming spatial query processing system is two-fold: (i) it should support the processing of a wide range of streaming spatial queries at runtime, and (ii) it should be extendable to accommodate user-defined streaming spatial queries. It is essential to have a generic query processing paradigm for streaming spatial query processing systems, as different queries have significantly different processing logic.

\squishend

In the literature, many spatial systems~\cite{alarabi2018summit, bakli2019hadooptrajectory, ding2018ultraman} are built on distributed data analytics frameworks (e.g, Hadoop and Spark) and designed to process large-scale static data, instead of real-time streaming data. 
The spatial extensions of RDBMSs~\cite{mongodb, spatialite} also process  streaming spatial queries but the query efficiency is low as the I/O bottleneck of these extensions is obvious~\cite{alam2021survey}.
These systems typically have higher latencies due to data being processed in batches or through complex query executions.
To better suit real-time scenarios, several spatial data processing engines built on streaming systems have been developed~\cite{shaikh2020geoflink, chen2021star, chen2020sstd}. 
These engines operate with low latency, allowing for the rapid ingestion and processing of data.
However, none of them enjoy excellent efficiency, good scalability, and nice generality simultaneously, as we will elaborate shortly.

To address all the challenges above, we propose \thor, a streaming system for real-time spatial data analytics.
The code of \thor{} is available at \cite{cheetahgiscode}.
Specifically, \thor{} is built upon Apache Flink Stateful Functions (i.e., \statefun{})~\cite{statefun}, which is an API for building distributed streaming applications with an actor-like model.
The actor-like model enables us to easily implement complex logic involving communications among multiple operators.
The architecture of \thor{} consists of several modules, i.e., \textsf{Transformer}, \textsf{Indexer}, \textsf{Local Processor}, \textsf{Aggregator}, \textsf{Load Balancer}, and \textsf{Metadata Synchronizer}.
The \textsf{Transformer} module processes data updates and answers simple queries (i.e., object query).
The \textsf{Indexer} module keeps a grid-based index, which serves two purposes: (i) pruning unqualified candidates for efficient query processing,
and (ii) assigning local processors to process queries for the unpruned candidates exactly.
Each \textsf{Local Processor} instance executes the queries with the partial data stored on it via a task slot.
The results of all \textsf{Local Processor} instances are aggregated by the \textsf{Aggregator} module to return the final query result.
Both \textsf{Load Balancer} and \textsf{Metadata Synchronizer} modules are devised to improve the overall performance of \thor.

Given the system components, \thor{} enjoys good scalability due to its flexible modular system architecture.
It supports many streaming spatial queries (e.g., object query, range count query, and $k$ nearest neighbor query), and can be easily extended to support user-defined queries as it provides a general query processing paradigm by utilizing its system modules.
Moreover, \thor{} achieves efficient query processing for two reasons:
(i) a suite of techniques and optimizations (e.g., grid-based global index, \textit{many-to-one} execution mode, and fine-grained resource management) are devised;
and
(ii) the \textsf{Load Balancer} and \textsf{Metadata Synchronizer} modules speed up  query processing by balancing the workloads in the distributed system and pruning the unqualified candidates efficiently.

The major technical contributions of this work are summarized as follows.

\begin{enumerate}
    \item We architect \thor, which provides a holistic solution to streaming spatial query processing with high throughput and low latency (see Section~\ref{sec:sys}).
    \item We propose a suite of techniques and optimizations to achieve good query processing performance in \thor{} (see Section~\ref{sec:dsg}).
    \item We devise a unified streaming spatial query processing paradigm in \thor{}, and elaborate how it can be applied to process three representative queries (i.e., object query, range count query, and $k$NN query) and user-defined queries (see Section~\ref{sec:qry}).
    \item We verify the superiority of \thor{} for processing various streaming spatial queries and the effectiveness of our proposed optimizations and techniques by extensive experimental studies on various datasets (see Section~\ref{sec:exp}).
\end{enumerate}

\begin{table*}
    \centering
    \label{tab:summary}
    \caption{A Summary of Spatial Analysis Systems }
    \begin{tabular}{|c|c|c|c|c|c|}
    \hline
       System & Framework & Real-time Support & Data Type & Load Balance & Category  \\ \hline
       SpatialLite~\cite{spatialite} & SQLite & No & Point, LineString, Polygon & - & Spatial database (Sec.~\ref{sec:database}) \\ \hline
       MongoDB~\cite{mongodb} & MongoDB & No & Point, LineString, Polygon & - & Spatial database (Sec.~\ref{sec:database}) \\ \hline
       SpatialHadoop~\cite{eldawy2015spatialhadoop} & Hadoop & No & Point, LineString, Polygon & - & Big data system (Sec.~\ref{sec:bigsys}) \\ \hline
       GeoSpark~\cite{yu2015geospark} & Spark & No & Point, LineString, Polygon & - & Big data system (Sec.~\ref{sec:bigsys}) \\ \hline
       Dragoon~\cite{fang2021dragoon} & Spark & Yes & Trajectory & Yes  & Big data system (Sec.~\ref{sec:bigsys}) \\ \hline
       SSTD~\cite{chen2020sstd} & Storm & Yes & Spatio-textual data & Yes & Streaming system (Sec.~\ref{sec:streamingsys})  \\ \hline
       GeoFlink~\cite{shaikh2020geoflink} & Flink & Yes & Time-window based events & Yes  & Streaming system (Sec.~\ref{sec:streamingsys}) \\ \hline
       {\bf Our \thor{}} & {\bf Flink Statefun} & {\bf Yes} & {\bf Moving objects} & {\bf Yes} & Streaming system with actor model \\ \hline
    \end{tabular}
\end{table*}

\section{Related work}\label{sec:bkg}
In this section, we introduce existing spatial data analysis systems, which are summarized in 
Table~\ref{tab:summary}.
Our \thor{} differs from them in terms of (i) the underlying framework and (ii) supported data type.
In addition, \thor{} naturally supports real-time queries and considers load-balance in its architecture.

\subsection{Spatial Data Analysis \& Spatial Databases}\label{sec:database}

With the advancement of location positioning tracking technologies~\cite{zheng2015trajectory}, many algorithms ~\cite{tao2013approximate, ayanso2014range, xu2018enabling, liu2021accurate,ouyang2020progressive, gu2016moving, li2016fast,tang2022discovering,tang2017efficient} have been devised to solve various spatial data analysis problems (e.g., range query, motif discovery, maximum range sum) 
in the literature. 
Almost all of them are designed to solve a specific query efficiently while we propose a streaming spatial query processing system to process various queries in this work.
From the system perspective, spatial databases have been developed decades ago. 
Some of them are extensions of RDBMS, e.g., PostGIS~\cite{postgis} and SpatialLite~\cite{spatialite}. 
There are also several NoSQL databases that provide spatial supports. 
Some of them have native spatial support (e.g., Redis~\cite{redis} and MongoDB~\cite{mongodb}) as they provide built-in optimizations for spatial data.
The rest of them use NoSQL databases as their underlying systems, e.g., GeoMesa~\cite{hughes2015geomesa} on Accumulo, JUST~\cite{li2020just} on HBase. 
Although these systems are more favorable than RDBMS when handling large-scale spatial data, they are not directly applicable to real-time applications.
In addition, they are not designed to avoid disk I/O, and thus suffer from performance degradation due to a large amount of data reads and writes.

\subsection{Big Data Spatial Analysis Systems}\label{sec:bigsys}
To accommodate the need for highly scalable distributed systems, several big spatial data processing systems are developed in recent years~\cite{alam2021survey}.  
Many of them (e.g., Summit~\cite{alarabi2018summit}, HadoopTrajectory~\cite{bakli2019hadooptrajectory}, TrajSpark~\cite{zhang2017trajspark}, UlTraMan~\cite{ding2018ultraman}, and Dragoon~\cite{fang2021dragoon}) are built on big data frameworks (e.g., Hadoop and Spark).
These systems have high scalability in terms of cluster resources.
Hadoop-based spatial data analysis systems have disk I/O bottlenecks while Spark-based systems achieve better performance because Spark is optimized to use large amounts of memory.
However, both types of systems are not fully applicable to real-time scenarios.
Spark-based systems still have high latency as they use micro-batch for data transfer~\cite{karimov2018benchmarking}.
Moreover, frequent modifications to the locations of the moving objects pose challenges to the immutable RDDs of Spark, albeit that UltraMan~\cite{ding2018ultraman} and Dragoon~\cite{fang2021dragoon} proposed techniques to alleviate the problem.

\subsection{Streaming Spatial Analysis Systems}\label{sec:streamingsys}

Streaming query processing has been widely used in many business applications, e.g., fraud detection, and online recommendation~\cite{carbone2020beyond,dayarathna2018recent}.
On one hand, many commercial streaming query processing systems have been developed in leading IT organizations (e.g., Twitter~\cite{kulkarni2015twitter,floratou2017dhalion}, Facebook~\cite{mei2020turbine}, and LinkedIn~\cite{noghabi2017samza}).
On the other hand, many excellent open-source streaming systems have been built by the community. The most representative ones include Apache Storm~\cite{storm}, Apache Spark Streaming~\cite{sparkstreaming}, and Apache Flink~\cite{flink}. Apache Flink has a better durable storage mechanism than Apache Storm.
In this work, we focus on streaming spatial query processing over a large amount of moving objects with a high update rate.
In particular, the spatial locations of the moving objects are generated and processed in a streaming manner to answer various spatial queries (e.g., $k$NN query, range query).

Several systems based on stream processing engines, such as Apache Storm~\cite{storm} and Apache Flink~\cite{flink}, have been developed for spatial data analysis.
Specifically, Tornado~\cite{mahmood2015tornado}, PS2Stream~\cite{chen2017ps2stream}, and SSTD~\cite{chen2020sstd} are Storm-based systems for Spatial-textual data streams. 
In these works, they attach additional textual keyword information to the spatial data and thus can support complex query conditions.
In addition, STAR~\cite{chen2021star} is specifically optimized for handling snapshot and continuous aggregate queries.
These systems also provide mechanisms to handle load balancing issues.
The above systems differ from our \thor{} as they focus on data models that are spatial-textual data and time-window based events, instead of moving objects in \thor{}, as we will elaborate it in Section~\ref{sec:preliminary}.
As a result, their system architectures and optimization techniques are not tailored for streaming spatial queries over a large amount of frequently updated moving objects. 

GeoFlink~\cite{shaikh2020geoflink} is built upon Flink, which processes window-based static spatial data stream.
It consists of the Filter phase and Refine phase.
It maps the input data to continuous query reducers for the registered queries (i.e., continuous spatial range query, $k$NN query and join query).
The generality of GeoFlink is low as it constructs the processing pipeline from scratch for different queries.
In addition, GeoFlink does not consider the load balancing issue, which results in poor scalability.
Several specialized streaming spatial systems have been studied in the database community.
For example, CarStream~\cite{zhang2017carstream} has been devised to develop safety-critical stream-processing applications using multiple types of driving data.
\cite{patroumpas2018fly} proposed an application framework to support on-the-fly mobility event detection over aircraft trajectories.
They differ from our work as they do not support various streaming spatial queries.

\section{Preliminary}\label{sec:preliminary}
In this section, we introduce the preliminaries of \thor{}, i.e., \statefun{} and the supported data model. 

\stitle{Stateful Functions in Flink}
\statefun{}~\cite{statefun} is an API to build stateful streaming applications~\cite{carbone2017stateful} with an actor-like model.
In particular, \statefun{} requires users to implement functions, which are the building bricks of the applications.
These functions can have parallel instances, message each other, and hold their own instance-level persisted states.
In addition, these functions may communicate with external systems to receive or send messages.
In this work, we employ \statefun{} as the underlying framework.
Compared to the static dataflow direct acyclic graph (DAG), which is commonly used in stream processing systems (e.g., Storm\cite{storm} and Flink~\cite{flink}),
functions in \statefun{} are more flexible as they can message each other arbitrarily, which significantly reduces the complexity of system design.
Moreover, \statefun{} is built upon \flink{} and inherits the acclaimed \flink{} features (e.g., exactly-once guarantee) in the stateful streaming field.
Besides, \statefun{} is built for the serverless architecture~\cite{serverless}, which handles resource management for the applications, and allows the developers to focus on the logic of their applications.
In particular,  \statefun{} randomly assigns each function instance to a task slot for execution, where task slot is an abstraction of a CPU and a fixed-size RAM in the cluster.

\stitle{Data Model}
To capture real world streaming scenarios, we focus on \textit{moving objects}~\cite{fang2021dragoon}, e.g., taxis, buses, and passengers, in \thor.
Each moving object $o_i$ sends its latitude and longitude $(x_t, y_t)$ at timestamp $t$ (e.g., collected via GPS devices) to the system continuously, and each data tuple is denoted as $(i, t, x_t, y_t)$.

\section{\thor{} Overview}\label{sec:sys}
In this section, we present the architecture overview of \thor{} and briefly introduce its modules.

Figure~\ref{fig:system_architecture} depicts the architecture of \thor.
\thor{} collects the data generated by moving objects continuously to support streaming query processing.
The end-users issue streaming spatial queries (e.g., kNN query, range query) via front-end applications.
Simple queries can be answered by the \textsf{Transformer} module directly, and we will introduce the details shortly.
The other queries will be processed by the \textsf{Indexer} module, which assigns each query to a set of \textsf{Local Processor} instances.
Each \textsf{Local Processor} instance computes the partial result of the query.
The partial results from all \textsf{Local Processor} instances are aggregated at the \textsf{Aggregator} module to obtain the query final result, which is returned to end-users. 
These four modules can be scaled-out to all machines in the cluster.
\thor{} also embeds two important modules, i.e., \textsf{Metadata Synchronizer} and \textsf{Load Balancer}, which significantly  improve the performance of query processing by managing index and balancing load, respectively.

\thor{} adapts a Map-Reduce-like architecture, which is widely used in many spatial systems \cite{chen2021star, chen2016realtime}.
This architecture could be scaled out to process a large number of queries easily.
However, it is still challenging to fully unleash the computational ability of the underlying clusters.
The reasons are three-fold.

\begin{enumerate}
    \item Each query will be assigned to a large number of \textsf{Local Processor} instances on the \textsf{Indexer}.
    \item Many sub-queries sharing the same \textsf{Local Processor} instances are sent separately, which incurs extra overhead on network transmission.
    \item The data and query may be skewed, which leads to workloads imbalance among different instances (e.g., local processor).
\end{enumerate}

As we will show in Section~\ref{sec:dsg}, \thor{} overcomes these issues by devising novel techniques. 
In the following, we briefly elaborate the functionality of each module in \thor.

\begin{figure}
    \small
	\centering
	\includegraphics[width=1.0\columnwidth]{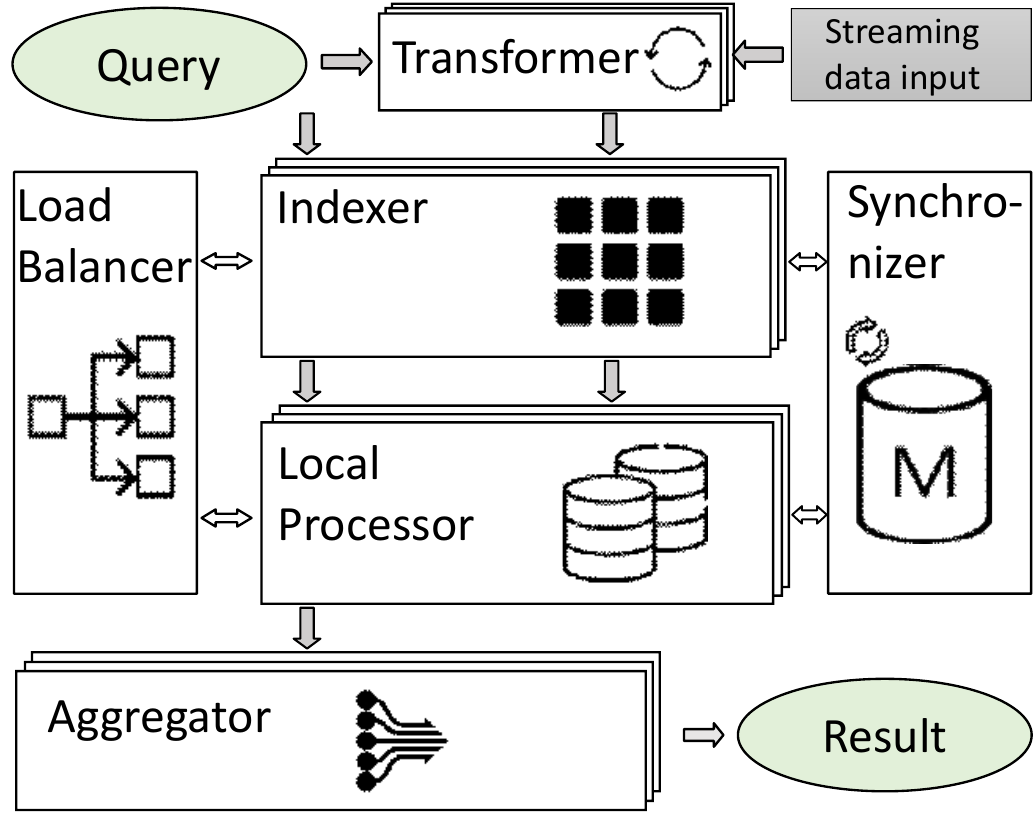}
	\caption{Architecture of \thor}
	\label{fig:system_architecture}
    \trim
\end{figure}

\stitle{Transformer module}
The location update record of an object does not contain its previous location, so \thor{} employs the \textsf{Transformer} module to transform the newest data of a moving object to a movement of the object.
We employ a hash table in \textsf{Transformer} to map the latest data tuple of each object.
For each collected data, the \textsf{Transformer} first finds the corresponding data tuple in the hash table.
Then, the hash table updates the latest location.
Finally, the \textsf{Transformer} shuffles the data tuples (i.e., object movement data) into the \textsf{Indexer} module for further processing.

\stitle{Indexer module}\label{sec:indexer}
The global index is maintained by the \textsf{Indexer} module in \thor.
The main responsibilities of the \textsf{Indexer} module are: (i) pruning cells to be searched for the queries, and (ii) assigning the data from \textsf{Transformer} module and queries to \textsf{Local Processors}.

\begin{figure}
    \small
	\centering
	\includegraphics[width=1.0\columnwidth]{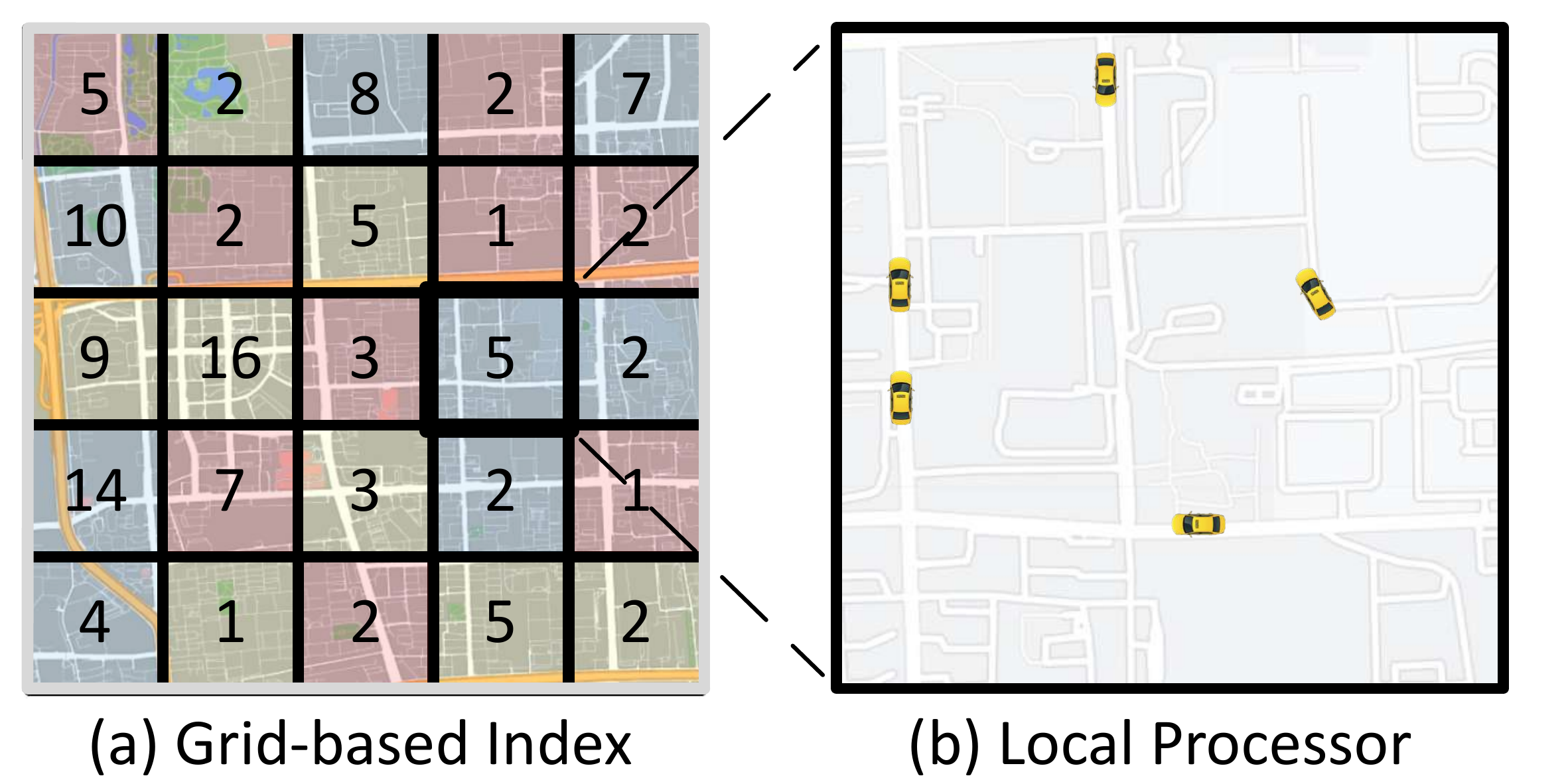}
	\caption{Grid-based index}
	\label{fig:grid_index}
    \trim
\end{figure}

\thor{} utilizes a grid-based index as global index in the \textsf{Indexer} module.
Compared to tree-based indexes, the grid-based index is simpler for data updates and query routing.
Figure~\ref{fig:grid_index}(a) illustrates the grid-based index.
It divides the entire region (e.g., a city, an administrative area) into equal-sized cells (a.k.a., grids).
Specifically, the grid-based index of the \textsf{Indexer} module only maintains the correspondence between each cell and its responsible \textsf{Local Processor} instance.
As shown in Figure~\ref{fig:grid_index}(a), cells with different colors correspond to different \textsf{Local Processor} instances.
The exact locations of the moving objects in each cell are stored at the corresponding \textsf{Local Processor} instance, as illustrated in Figure~\ref{fig:grid_index}(b).
For example, the locations of the five moving objects in the blue cell in Figure~\ref{fig:grid_index}(b) are stored in a \textsf{Local Processor} instance, along with the other moving objects from the other blue cells.

For data processing, the \textsf{Indexer} module treats the movement of an object as applying two basic operators (\textsf{delete} and \textsf{insert}) in the 
\textsf{Local Processor} instances.
The \textsf{Indexer} module identifies the corresponding \textsf{Local Processor} instance of each object by the location of its associated data  and routes the object to the instance accordingly.
For query processing, the \textsf{Indexer} module applies the metadata (i.e., the total number of moving objects) of each cell, as the values of each cell shown in Figure~\ref{fig:grid_index}(a),  to prune the unqualified cells at first, then routes the query to the \textsf{Local Processor} instances, which store the data of the unpruned cells, to compute the final result.

To sum up, the grid-based global index holds the following nice properties.
First, the cell size is predefined via system configuration.
It has two advantages: (i) it is tunable and (ii) it incurs light-weight memory consumption as it only has a constant number of cells.
Second, the time cost to find the \textsf{Local Processor} instance for each moving object is $O(1)$ as the cells have equal-size.
Third, it provides fast-and-loose pruning for query processing in \thor. We will explain how it works for query processing in 
Section~\ref{sec:qry}.

\stitle{Local Processor module}
In \thor, \textsf{Local Processor} stores the exact locations of moving objects that belong to the same cell, as shown in Figure~\ref{fig:grid_index}(b).
Moreover, it processes the received queries accordingly.
A hash table is employed as the local index in each \textsf{Local Processor} instance to offer high throughput and efficient data migration.
We optimize the communication cost by designing execution mode on \textsf{Local Processor}, which will be introduced in Section~\ref{sec:cellgroup}.

\stitle{Aggregator module}
\thor{} aggregates the partial query results, which are sent by different \textsf{Local Processor} instances, and returns the final result to the users by \textsf{Aggregator} Module.
It is important to determine whether the partial results from every \textsf{Local Processor}  instance are received in the \textsf{Aggregator} module for distributed streaming query processing.
\thor{} achieves that by attaching the total number of incurred \textsf{Local Processor} instances to each routed query in the \textsf{Indexer} module.

\stitle{Metadata Synchronizer module}
The \textsf{Metadata Synchronizer} module synchronizes the meta-data of each cell in the \textsf{Local Processor} instance to the grid-based global index in \textsf{Indexer} module.
It plays a vital role to improve the efficiency of \thor{} when the data are deleted and inserted frequently in each cell.
We will introduce its details in Section~\ref{sec:meta}.

\stitle{Load Balancer module}
The \textsf{Load Balancer} module monitors the workload of each \textsf{Local Processor} instance and resolves imbalance issues when  necessary, e.g., the workloads of two \textsf{Local Processor} instances are very different, as we will show in Section~\ref{sec:loadbalance}.

\section{\thor{} Designs}\label{sec:dsg}
In this section, we present the key technical contributions in \thor{}.

\subsection{Fine-Grained Resource Management}\label{sec:rm}
As introduced in Section~\ref{sec:bkg}, \statefun{} is a serverless architecture, and it automatically manages hardware resources.
However, automatic resource management does not make full utilization of the resources for streaming spatial query processing.
In particular, \flink{} randomly allocates resources to each function instance.
According to our empirical experiments, it causes severe resource preemption among the working set of \statefun{} instances during query processing.

To address this problem, \thor{} provides fine-grained resource management via revising the task slot allocation scheme of \flink{}.
In particular, \thor{} gives an option to the developers, which allows them to allocate an exclusive task slot to a certain \statefun{} instance.
With fine-grained resource management in \thor, \statefun{} instances on \thor{} can be created on task slots with efficient resource utilization.
For example, \textsf{Metadata Synchronizer} and \textsf{Load Balancer} are allocated exclusive task slots in \thor{} in order to avoid unnecessary overheads, which are incurred by task slot preemption with the \textsf{Local Processor} instances.

\subsection{Local Processor Execution Mode}\label{sec:cellgroup}
The global index (i.e., grid-based index) in \thor{}  divides the geographic region into cells.
For the \textsf{Local Processor} execution mode, a straightforward solution is mapping each cell to a \textsf{Local Processor}, denoted as \textit{one-to-one} mode,
and then each \textsf{Local Processor} instance is instantiated and executed in a task slot.
However, this solution has three major issues.
First, the number of cells in the grid-based index can be tens of thousands or even more. 
Thus, it is impractical to assign each cell to a \textsf{Local Processor} as it incurs expensive overhead to schedule every \textsf{Local Processor} instance during query processing.
Second, the \textsf{Local Processor} instances of a query are executed by a fixed set of task slots (determined by the configurations of the physical computing cluster), which incur severe resource contention.
Third, the spatial locality of adjacent cells is ignored in \textit{one-to-one} mode, which can be exploited to accelerate query processing.

\begin{figure}
    \small
	\centering
	\includegraphics[width=0.85\columnwidth]{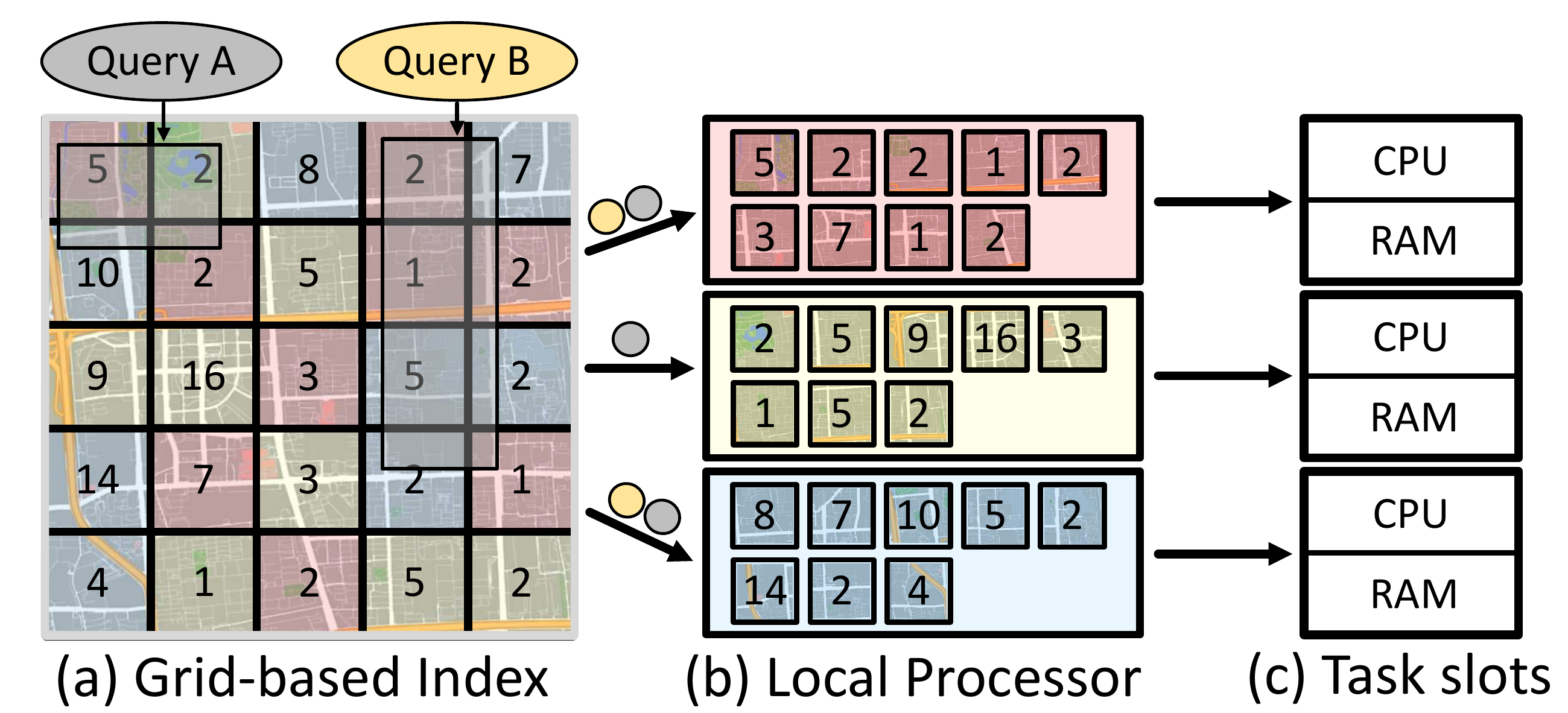}
	\trim
	\caption{Many-to-one execution mode of \textsf{Local Processor}}
	\label{fig:combiner}
    \trim
\end{figure}

To address the above limitations, we devise a \textit{many-to-one} execution mode in \thor.
Specifically, we assign multiple cells to one \textsf{Local Processor}.
As illustrated in Figure~\ref{fig:combiner}, the cells with the same background color are assigned to the same \textsf{Local Processor}.
There are three \textsf{Local Processors}, as shown in Figure~\ref{fig:combiner}(b), for all the cells in the grid-based index in Figure~\ref{fig:combiner}(a).
During query processing, \thor{} routes a query and the corresponding cells to these \textsf{Local Processors}.
For example, Figure~\ref{fig:combiner} depicts that Query A is routed to three \textsf{Local Processors}, and Query B is routed to two \textsf{Local Processors}.
These \textsf{Local Processors} are instantiated to \textsf{Local Processor} instances and are executed by the scheduled task slot in \thor.

The effectiveness of the \textit{many-to-one} execution mode of \textsf{Local Processor} are: (i) it exploits the spatial and temporal localities of the moving objects, as the cells in the same \textsf{Local Processor} can be processed in one batch;
and
(ii) the number of \textsf{Local Processors} is configurable, which can be set according to the number of task slots in the physical computing cluster.
It effectively reduces resource contention (e.g., task slots, network bandwidth).
We will evaluate its effectiveness in Section~\ref{sec:exp}.

\subsection{Metadata Synchronizer}\label{sec:meta}
In \thor, we design \textsf{Metadata Synchronizer} to update the metadata of every \textsf{Local Processor} in \textsf{Indexer}.
In particular, each \textsf{Local Processor} only applies the assigned operators (e.g., data deleting/inserting) on its cells, and thus it does not have the global information of other \textsf{Local Processors}.
However, the metadata of every cell in the grid-based index is frequently used during query processing (e.g., pruning unqualified candidate cells, and balancing workload).

The challenges to design an ideal \textsf{Metadata Synchronizer} are two-fold.
First, the streaming updates of massive moving objects change the metadata in each cell rapidly.
For example, the object count of a cell changes dramatically as many moving objects frequently moves from a cell (i.e., delete) to its adjacent cell (i.e., insert).
Second, the cost to synchronize metadata from \textsf{Local Processors} to \textsf{Indexer} is not negligible.
Thus, it is impractical to provide the latest metadata of each cell in \textsf{Indexer} under a high data update rate.

We devise the \textsf{Metadata Synchronizer} to address these challenges in \thor{}.
Specifically, it collects and integrates metadata from \textsf{Local Processors}, and broadcasts them to \textsf{Indexers}. 
To avoid resource contention, we assign an exclusive task slot to the instance of \textsf{Metadata Synchronizer}.
To reduce the network transmission cost, the \textsf{Metadata Synchronizer} collects metadata from the \textsf{Local Processors} in two manners.
\squishlist
\item \sstitle{Periodic synchronization} If the metadata of the cell in a \textsf{Local Processor} is outdated (e.g., exceeds a time threshold), the \textsf{Local Processor} will send the metadata of this cell to the \textsf{Metadata Synchronizer}.
\item \sstitle{Accumulative synchronization} If the value of a metadata changes significantly (e.g., exceed a pre-defined threshold), the \textsf{Local Processor} also will send it to \textsf{Metadata Synchronizer}.
\squishend

For metadata collected by periodic synchronization, the \textsf{Metadata Synchronizer} integrates them into a hash table, and periodically broadcasts the metadata of the cells that have been changed since the last broadcasting, to all \textsf{Indexers}.
For metadata collected by accumulative synchronization, \textsf{Metadata Synchronizer} immediately broadcasts it to all Indexers after it is received.
In addition, it is important to guarantee the freshness of the metadata in \textsf{Indexer}, as out-of-date metadata will cause performance degradation during query processing.
We propose metrics to measure the freshness of the metadata in \textsf{Indexers}, and allow users to set tunable thresholds for them.

\subsection{Load Balancer}\label{sec:loadbalance}
In distributed computing systems, load imbalance issues have significant impacts on system performance.
Thus, it is important to maintain balanced workloads among different working machines in \thor.
In particular, possible load imbalance issues in \thor{} include:
(i) the queries concentrated on some hot cells in the grid-based index,
and
(ii) the compute workload of \textsf{Local Processor} instances are quite different.
To alleviate the imbalance issues in \thor, we use random shuffle as the basic load balancing scheme, by  following~\cite{ramakrishnan2012balancing}.
Specifically, \textsf{Transformer} first sends data to a random \textsf{Indexer} instance, then the \textsf{Indexer} instance assigns the cells in the grid-based index to different \textsf{Local Processor} randomly, and the partial results of \textsf{Local Processor} instances aggregated by \textsf{Aggregator} instance via random scheduling task slot.

Even with the random shuffle, the workload bias still exists in \thor{} as the moving objects are frequently moving among these cells in the grid-based index.
Thus, we design a \textsf{Load Balancer} module in \thor, to alleviate the balance bias issue dynamically.
Specifically, it first monitors the imbalance among all \textsf{Local Processor} instances by defining a workload imbalance degree metric, then it immediately remedies the imbalance \textsf{Local Processor} instances when the workload imbalance metric value is larger than a given threshold.
The core idea of imbalance remedy is moving cells from high workload \textsf{Local Processor} to low workload \textsf{Local Processor}.

\subsubsection{Workload Metric of Local Processor Instance}
Suppose the cell group in \textsf{Local Processor} instance $p$  is $C_p$,
We first define the workload of each cell $c$ in $C_p$ as the number of computation operators (e.g., data updates and query processes) in it for a given time period $\Delta T$.
Specifically, the workload for each cell $c \in C_p$ is measured as
$$
W(c) = \frac{|U(c)| + \sum_{q \in Q(c)}N(q, c)}{\Delta T},
$$
where $U(c)$ is the set of data updating operators (i.e., \textsf{delete} and \textsf{insert}) received by $c$ in the time period $\Delta T$.
$Q(c)$ is the set of queries, which processed the moving objects in cell $c$, in the time period $\Delta T$, and $N(q,c)$ is the number of visited moving objects stored in $c$ during the processing of query $q$.

We then measure the workload of \textsf{Local Processor} instance $p$ by taking the summation of every cell's workload in it $W(p) = \sum_{c \in C_p}W(c)$.
In \thor, the \textsf{Load Balancer} model measures the workload on each \textsf{Local Processor} instance periodically, i.e., set $\Delta T$ as a constant value.
However, the \textsf{Load Balancer} model is customizable two-fold: (i) the workload metric of each \textsf{Local Processor} instance,  and (ii) the time period $\Delta T$ in each instance.

\subsubsection{Imbalance Degree among All \textsf{Local Processor} Instances}
The imbalance degree among all \textsf{Local Processor} Instances is defined as follows.
$$
Degree(P) = \frac{1}{|P|}\sum_{p \in P}(W(p) - \overline{W(P)})^2,
$$
where $P$ is the set of all \textsf{Local Processor} instances and $\overline{W(P)}$ is the average workload of all $W(p)$, i.e., the workload of every \textsf{Local Processor} instances.

If the value of imbalance degree $Degree(P)$ exceeds the pre-defined system imbalance bias tolerance value, the \textsf{Load Balancer} model will incur an imbalance remedy strategy, as we introduce shortly, to alleviate it.

\subsubsection{Hardness Analysis of Imbalance Remedy Problem}
\thor{} remedies the imbalance on \textsf{Local Processor} instances by moving the cells in heavy workload \textsf{Local Processor} instances to the light workload one(s).
We first define the imbalance remedy problem in Problem~\ref{prob:remedy}.

\begin{problem}\label{prob:remedy}
Given a set of \textsf{Local Processor} instances $P$, and each \textsf{Local Processor} instance $p \in P$ contains a group of cells.
The workload of each cell $c \in p$ is $W(c)$.
The imbalance remedy problem is finding a imbalance remedy plan $\mathcal{M}$ which transforms $P$ to $P^\prime$ such that the imbalance degree of $P^\prime$ is lower than a threshold $\theta$,
i.e., $Degree(P^{\prime}) \leq \theta$.
\end{problem}

The hardness of Problem~\ref{prob:remedy} is illustrated in Lemma~\ref{lem:nphard}.

\begin{lemma}\label{lem:nphard}
The imbalance remedy problem (Problem~\ref{prob:remedy}) is NP-hard.
\end{lemma}

\begin{proof}
It can be proved by reducing the well-known PARTITION problem~\cite{garey1979computers}, which is NP-hard, to the imbalance remedy problem.

\stitle{PARTITION problem} Consider $S$ is a set of numbers, the decision version of PARTITION problem returns whether the numbers can be partitioned into two sets $A$ and $A^\prime = S - A$ such that $\sum_{a \in A} a = \sum_{a \in A^\prime} a$.

We simplify imbalance remedy problem (Problem~\ref{prob:remedy}) by setting (i)  $P=\{p_1, p_2\}$ and (ii) $\theta = 0$. The simplified problem asks whether there exists a imbalance remedy plan to make $Degree(P') = 0$, i.e., $\sum_{c \in C_{p_1}}W(c) = \sum_{c \in C_{p_2}}W(c)$.

Given an instance of the PARTITION problem, we construct an instance of the simplified imbalance remedy problem as follows. First, set \textsf{Local Processor} instance $p_1 = \emptyset$, Second, for each number $a \in S$, insert a cell $c$ into \textsf{Local Processor} instance $p_2$ with $W(c) = a$.
The PARTITION problem is equivalent to a simplified imbalance remedy problem.
Thus, the simplified imbalance remedy problem is also NP-hard.
\end{proof}

\subsubsection{Imbalance Remedy Strategy}
The challenge of the imbalance remedy strategy is how to minimize the imbalance degree of the whole system.
In addition, the cell movement during the imbalance remedy process should not occupy too much network bandwidth to guarantee its efficiency.
We propose a greedy algorithm (see Algorithm~\ref{alg:remedy}) to remedy the imbalance issue heuristically.
The general idea of the greedy algorithm is as follows.
We first sort the \textsf{Local Processor} instances by the descending order of their workload metric value.
Then, we consider every \textsf{Local Processor} instance $p$ as a high workload instance and try to move some of its cells out to decrease the imbalance degree of the whole system.
Specifically, for each cell $c$ in the \textsf{Local Processor} instance, we move it to the \textsf{Local Processor} instance which currently has the minimum workload,
as the top element in the minimum heap $\mathcal{H}$ in Line~\ref{ln:heap}.
After that, we obtained a imbalance remedy plan $\mathcal{M}_{tmp}$ for the considering \textsf{Local Processor} instance $p$.
We next compute the imbalance degree of all \textsf{Local Processor} instances after moving all cells in the imbalance remedy plan $\mathcal{M}_{tmp}$,
and maintain the best plan we find so far $\mathcal{M}$, see Lines~\ref{ln:ifs} to ~\ref{ln:ife} in Algorithm~\ref{alg:remedy}.
The best imbalance remedy plan $\mathcal{M}$ is returned and applied by the \textsf{Load Balancer}.

\begin{algorithm}
 \caption{\textsf{ImbalanceRemedy}($P$)}\label{alg:remedy}
 \begin{algorithmic}[1]
    \State sort $P$ by the descending order of $\forall p \in P, W(p)$
    \State initialize imbalance degree $bsf \leftarrow 0$
    \State initialize imbalance remedy plan $\mathcal{M} \leftarrow \emptyset$
    \For {each $p$ in $P$}
        \State initialize the imbalance remedy plan $\mathcal{M}_{tmp} \leftarrow \emptyset$
        \State create min-heap $\mathcal{H}$ for $p_i \in P - \{ p \}$ by its $W(p_i)$
        \State push $\langle 0, p \rangle$ into $\mathcal{H}$
        \For{each cell $c$ in the cell group of $p$}
            \State $\langle w, tmp_p \rangle \leftarrow \mathcal{H}.pop()$ \label{ln:heap}
            \State push  $\langle w+W(c), tmp_p \rangle$ into $\mathcal{H}$
            \State insert $\langle c, p, tmp_p \rangle$ into $\mathcal{M}_{tmp}$
        \EndFor
        \State $\mathcal{P}_{tmp} \leftarrow$ all instances in $\mathcal{H}$
        \If{$bsf < Degree(P) - Degree(P_{tmp}) $} \label{ln:ifs}
            \State $bsf \leftarrow Degree(P) - Degree(P_{tmp})$
            \State $\mathcal{M} \leftarrow \mathcal{M}_{tmp}$
        \EndIf \label{ln:ife}
    \EndFor
    \State return $\mathcal{M}$
 \end{algorithmic}
\end{algorithm}

\begin{figure}
    \small
	\centering
	\includegraphics[width=0.90\columnwidth]{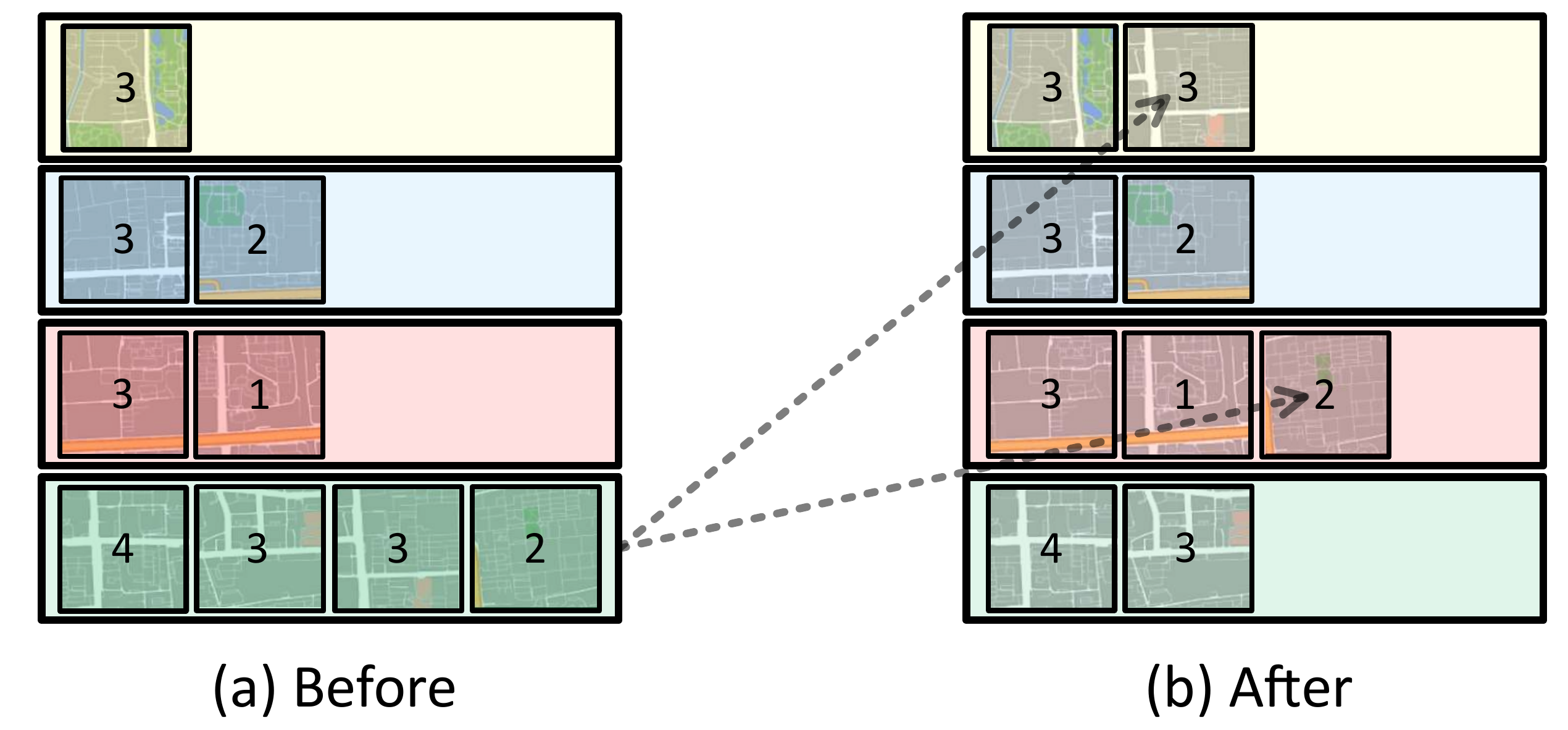}
	\trim
	\caption{Imbalance remedy example}
	\label{fig:irexample}
    \trim
\end{figure}

We illustrate the effect of our heuristic imbalance remedy algorithm (in Algorithm~\ref{alg:remedy}) by the concrete example in Figure~\ref{fig:irexample}.
Figure~\ref{fig:irexample} shows the set of \textsf{Local Processor} instance $P$ before we invoke the imbalance remedy algorithm.
The value in each cell illustrates its workload.  The imbalance degree of $P$ is 12.5 as the workload of the last \textsf{Local Processor} instance (it consists of 4 cells) is much larger than the rest instances.
Applying the above greedy remedy algorithm, the cells with workloads 3 and 2 will be moved out to the first and third \textsf{Local Processor} instances, respectively, as the arrow shown in Figure~\ref{fig:irexample}.
As a result, the imbalance degree of the whole system turns to 0.5.

\subsubsection{Cell Movement Method of Imbalance Remedy Plan}
A straightforward method to move the cells in the imbalance remedy plan is suspending all the input of data moving and queries until the cell movement completes.
However, it is impractical as it severely affects the overall latency and throughput in \thor{}.
To address this, we devise a cell-moving method that does not require halting the query processing, while ensuring the correctness of data operations and queries during the imbalance remediation phase.

Specifically, the cell $c$ in the moving-out \textsf{Local Processor} instance $p$ are packaged and sent to the moving-in \textsf{Local Processor} instance $p_{tmp}$.
In particular, \thor{} does not delete the data objects in $c$ at instance $p$ immediately,
it will delete all these moving-out cells after the cell moving phase completes in it.
In addition, we fixed the size of each package (e.g., each package only includes 10,000 data objects) during the cell moving process to avoid too long data movement in \thor.
The cell moving phase completes when all tuples (i.e., $\langle c, p, p_{tmp} \rangle$) in imbalance remedy plan $\mathcal{M}$ are processed.

After cell movement completed, the \textsf{Load Balancer} broadcasts the new \textsf{Local Processor} of moved cells to all \textsf{Indexers}, so that the grid-based index in every \textsf{Indexer}  map the moved cell to its \textsf{Local Processor} instance correctly.
The rebalance remedy procedure is finished after all \textsf{Indexers} are updated.
We suggest that the period to invoke \textsf{Load Balancer} model should not be too frequent
as it incurs slight overhead during the whole imbalance remedy procedure.

\subsubsection{Discussion}
Load balancing is an important and hot research topic in the distributed systems community. Numerous studies have been conducted in the literature~\cite{kolb2012load, ramakrishnan2012balancing, bindschaedler2018rock, daghistani2021swarm, fang2021dragoon}. 
Some of these methods may not be a good choice for \thor{} due to the overhead associated with data repartitioning and index modification when the system experiences frequent data updates.
Others may serve as alternative methods for remedying imbalance in \thor{}. We leave the detailed adoption and comparison of these methods as future work.

\section{Query Processing}\label{sec:qry}
In this section, we first illustrate the query processing procedure in \thor. 
Then we present its query processing paradigm to confirm the generality of \thor{}. 

\subsection{Typical Query Processing Procedure} 

\stitle{Object query}
The \textsf{Transformer} module can be used to answer a special kind of query, i.e., object query, directly.
object query (see Definition~\ref{def:objqry}) is a simple query, which usually is used to find the location of a moving object (e.g., a car or a bus) in real-world applications.

\begin{definition}\label{def:objqry}
Given the moving object ID $o_i$, the object query returns the latest location of $o_i$.
\end{definition}

The query processing procedure of object query is that: (i) \thor{} received the object query with moving object id $o_i$, and shuffles the query to the \textsf{Transformer} module as it maintains a hashing table, which records the latest position of the moving objects in \thor.
The query is answered by simply looking up the hash table in \textsf{Transformer} and returns the location as the result.

\stitle{Range count query}
The range count query (see Definition~\ref{def:rangequery}) returns the number of moving objects in a given region, it is widely used to analyze the traffic jam degree in intelligence transportation applications.

\begin{definition}\label{def:rangequery}
Given a query region $r$, the region count query returns the number of moving objects whose latest location lies in region $r$.
\end{definition}

\begin{figure}
    \small
    \centering
    \begin{tabular}{cc}
    \includegraphics[width=0.30\columnwidth]{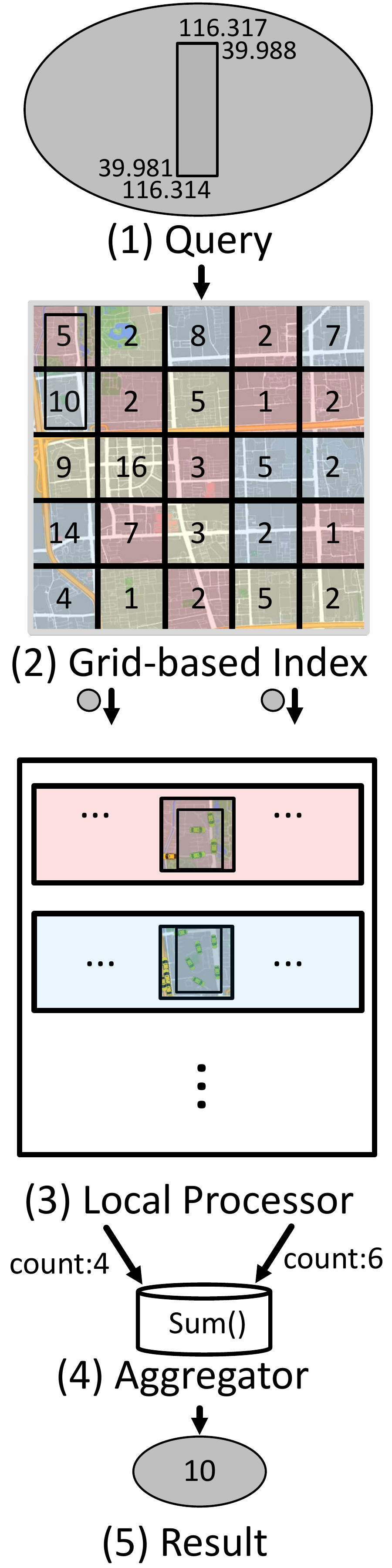}
    &
    \includegraphics[width=0.30\columnwidth]{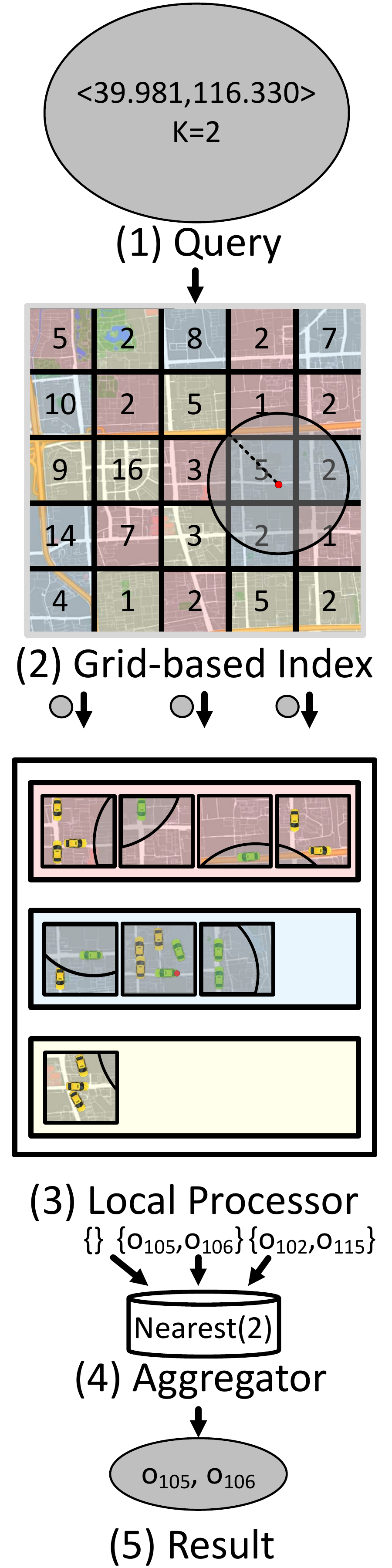}
    \\
    (a) range count query
    &
    (b) $k$NN query
    \end{tabular}
    \caption{Query processing in \thor}
    \label{fig:query}
    \trim
\end{figure}

The range count query processing procedure is illustrated in Figure~\ref{fig:query}(a).
Give a query region $r$, as Figure~\ref{fig:query}(a)(1) shown, the \textsf{Indexer} module first finds the intersection cells of $r$, both two intersected cells are illustrated in Figure~\ref{fig:query}(a)(2).
They belong to two different \textsf{Local Processor} instances, as shown in their background colors.
Thus, \thor{} instantiates two \textsf{Local Processor} instances to process the region count query with its underlying cells, see Figure~\ref{fig:query}(a)(3).
\thor{} schedules task slots to execute the above two \textsf{Local Processor} instances and return the local count result to the \textsf{Aggregator} model in Figure~\ref{fig:query}(a)(4),
and returns 10 as the final result of the given region count query with $r$.

\stitle{$k$ Nearest Neighbor Query}
The $k$ nearest neighbor query (see $k$NN query in Definition~\ref{def:knnquery}) finds the $k$ nearest moving objects to a given location.
It is the most frequent operator in ride-sharing services, e.g., Uber, and Didi.

\begin{definition}\label{def:knnquery}
Given a query location $q$ and an integer $k$, the $k$NN query returns the locations of $k$ objects whose distance to $q$ are top-$k$ smallest among all objects.
\end{definition}

Figure~\ref{fig:query}(b) shows an example. On \textsf{Indexer} module, \thor{} prunes the unqualified cells by exploiting the metadata of each cell.
Specifically, \thor{} identifies a minimum circle region, which guarantees it includes at least $k=2$ moving objects for the giving location $q$.
As shown in Figure~\ref{fig:query}(b)(2), \thor{} disqualifies all other cells except the 8 cells which are covered by the minimum circle region.
These 8 cells belongs to  3 \textsf{Local Processors}, as illustrated in Figure~\ref{fig:query}(b)(3).
In each \textsf{Local Processor} instance, it finds the top-$2$ nearest moving objects of query $q$ with the considering cells.
The local result of each \textsf{Local Processor} instance are aggregated in \textsf{Aggregator} module, see Figure~\ref{fig:query}(b)(4),
to return the final top-$2$ nearest moving objects (i.e., $o_{105}$ and $o_{106}$) in Figure~\ref{fig:query}(b)(5).

\stitle{Continuous query processing}
All the queries mentioned above are snapshot queries as one query result is immediately returned after the query is processed.
In addition, \thor{} also supports continuous query \cite{chen2021star}, which means users can register a continuous query, and then the query results are continuously refreshed. 

\subsection{Query Processing Paradigm in \thor{}}
With the above three representative queries, we summarize the query processing paradigm in the following steps.

First, \thor{} receives the user issued query, and processes it in \textsf{Transformer} if necessary. 
Then, \thor{} incurs \textsf{Indexer} to prune the unqualified cells of the query, then routes the query with candidate cell group to \textsf{Local Processors}.
Next, \thor{} computes the local result by scheduling every \textsf{Local Processor} instances to task slots.
Finally, \thor{} aggregates the local results in each \textsf{Local Processor} instance and returns the final result of the query.

In addition, the end-users also can issue user-defined queries by extending \thor{} by interpreting the queries in above steps.
Specifically, \thor{} abstracts the query processing steps into several functions, 
which can be combined to process the user-defined queries.

\section{Empirical Evaluation}\label{sec:exp}
In this section, we demonstrate the superiority of \thor{} via extensive experiments evaluation.

\subsection{Experimental Setting}\label{sec:expsetting}

\begin{table}
\small
\centering
    \caption{The statistics of datasets}\label{table:dataset}
    \trim
    \begin{tabular}{|c | c | c | c|}      \hline
      & \textsf{Shopping} & \textsf{Geolife} & \textsf{Brinkhoff} \\   \hline\hline
     \makecell{\# of Moving\\ objects} & 111,488 & 17,784 & 1,000,000 \\     \hline
     \# of Updates & 307,679,087 & 24,876,978 & 362,720,393\\    \hline
     Raw size & 19GB  & 1.4GB & 27GB\\   \hline
     \makecell{Space area\\ size}& \makecell{91 $\times$ 56\\ (m)} &  \makecell{Lon.:-149.8\textasciitilde174.3\\Lat:1.2\textasciitilde61.2} & \makecell{12 $\times$ 14\\(km)}\\   \hline
     \makecell{Default cell\\ size} & \makecell{$1.0 \times 0.667$ \\(m)} &  $5 \cdot 10^{-4} \times 5 \cdot 10^{-4}$ & \makecell{$5 \times 5$ \\(m)} \\      \hline
    \end{tabular}
\end{table}

\stitle{Dataset and query set} In this work, we use three different datasets (see Table~\ref{table:dataset}), which have a wide range of moving objects characteristics (e.g., the number of moving objects, the number of object updates, the space area size), to evaluate the performance of query processing in our proposal \thor{}.
In particular, \textsf{Shopping}\cite{shopping} includes the trajectories of visitors at the ATC shopping center in Osaka.
\textsf{Geolife}\cite{zheng2010geolife} consists of the trajectories of users, which are collected by Microsoft Research (Beijing).
\textsf{Brinkhoff}\cite{brinkhoff} is a synthetic dataset, which generates the trajectories on the road network of Las Vegas via the Brinkhoff generator. We follow~\cite{ding2018ultraman}\cite{fang2021dragoon} and use it to verify the scalability of our proposal.
We use \textsf{Shopping} as the default dataset, and the default cell size of the grid-based index for each dataset is shown in Table~\ref{table:dataset}.
We generate the query set of each dataset by following the distribution of the data points.
Specifically, we divide the whole space area into small regions and count the data points in each region.
Then we generate the queries in each region, and the number of generated queries in it is proportional to its data points.

\begin{figure*}
    \centering
    \begin{tabular}{ccc}
    \includegraphics[width=0.63\columnwidth]{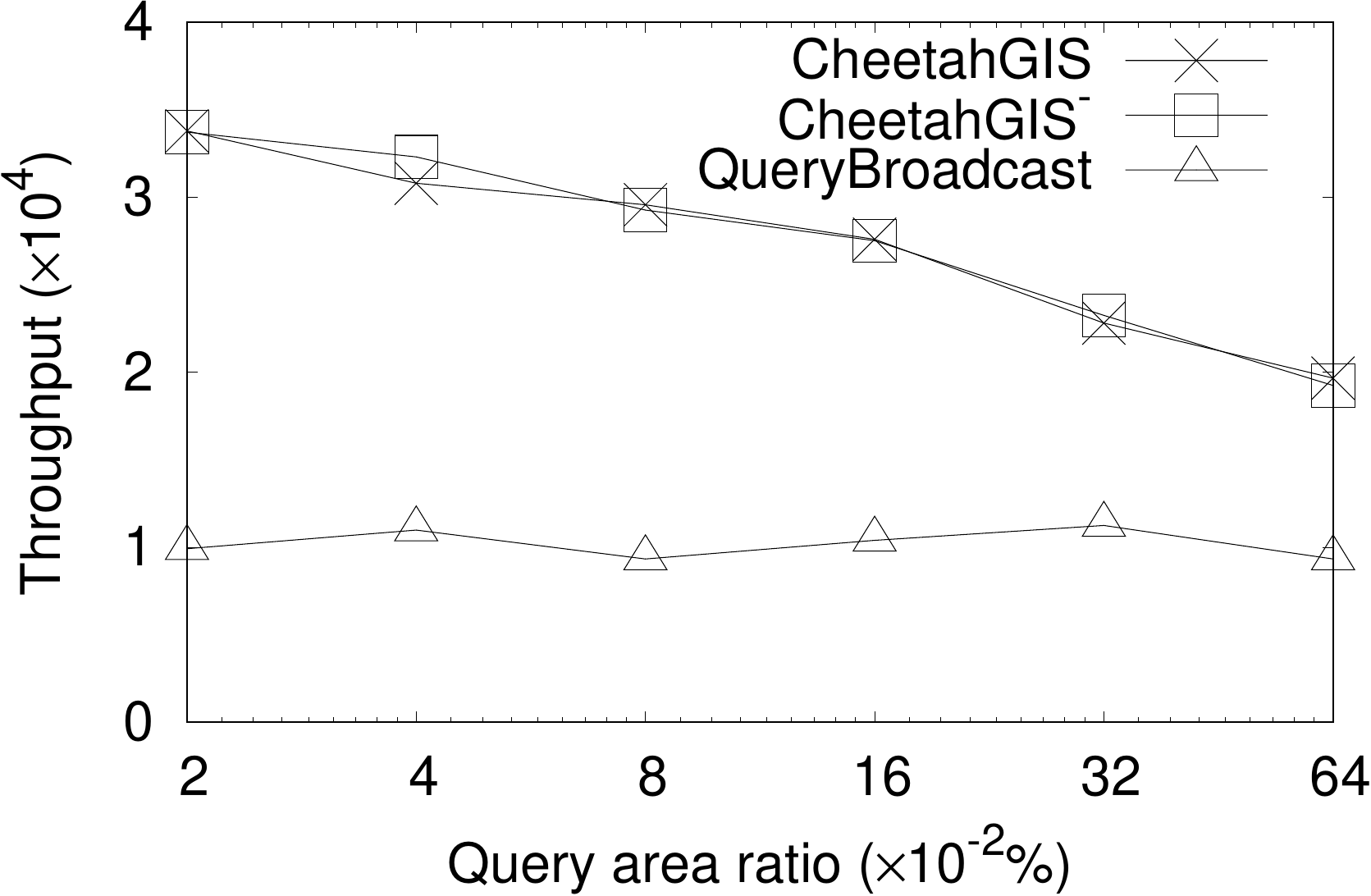}
    &
    \includegraphics[width=0.63\columnwidth]{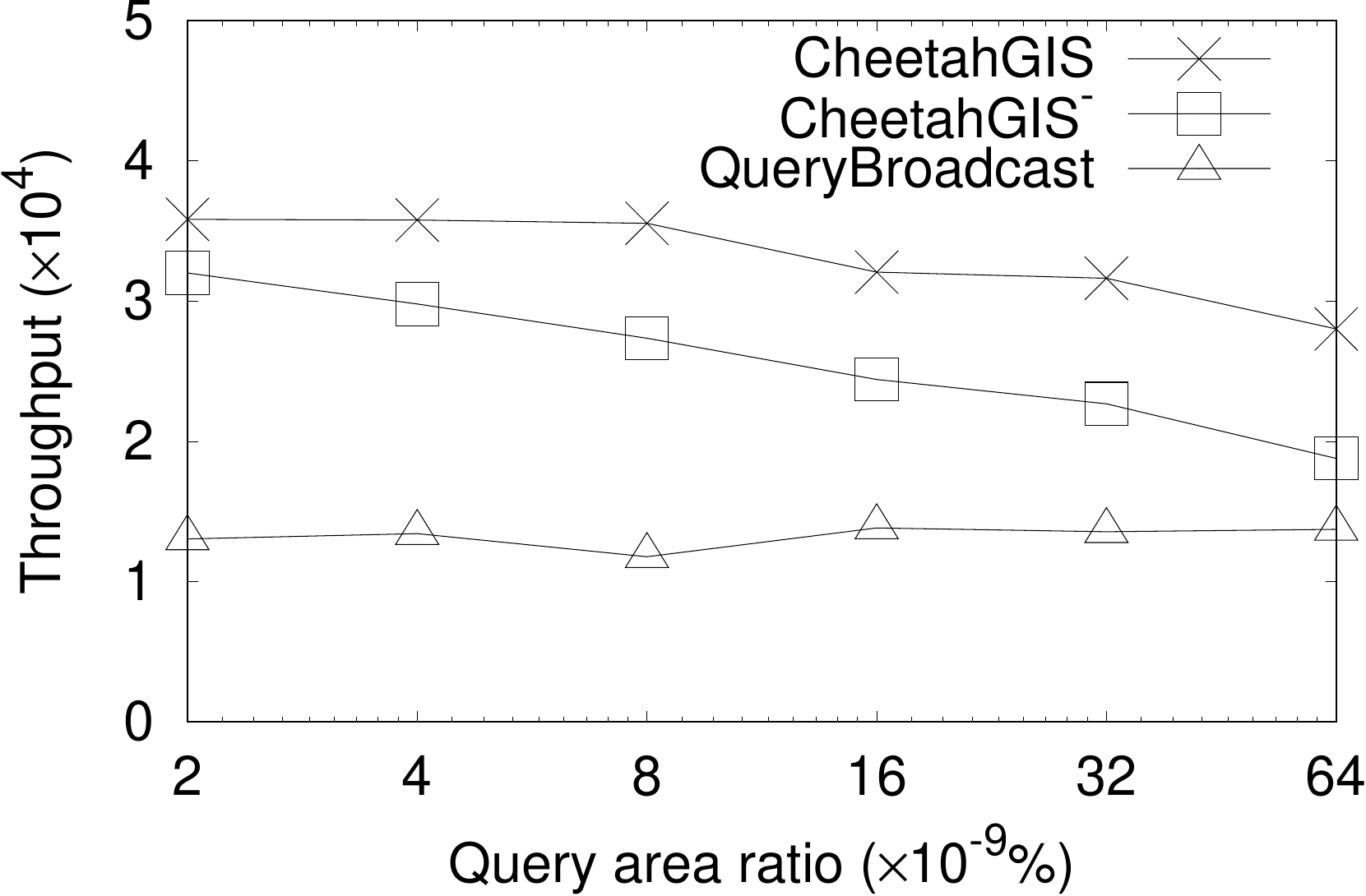}
    &
    \includegraphics[width=0.63\columnwidth]{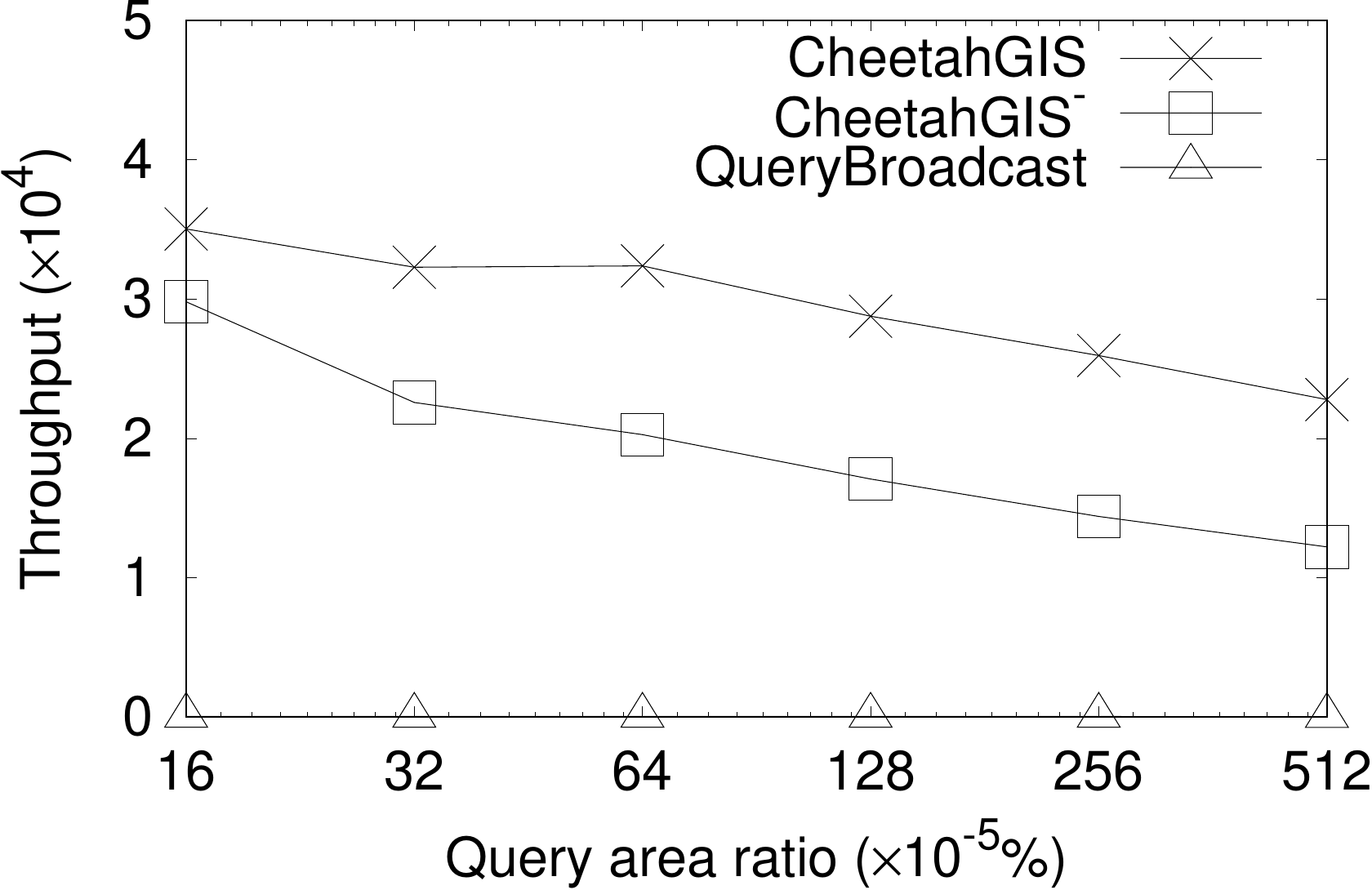}
    \\
    (a) \textsf{Shopping}
    &
    (b) \textsf{Geolife}
    &
    (c) \textsf{Brinkhoff}
    \end{tabular}
    \caption{Range count query performance evaluation, default cell size in each dataset}
    \label{fig:rngeva}
    \begin{tabular}{ccc}
    \includegraphics[width=0.63\columnwidth]{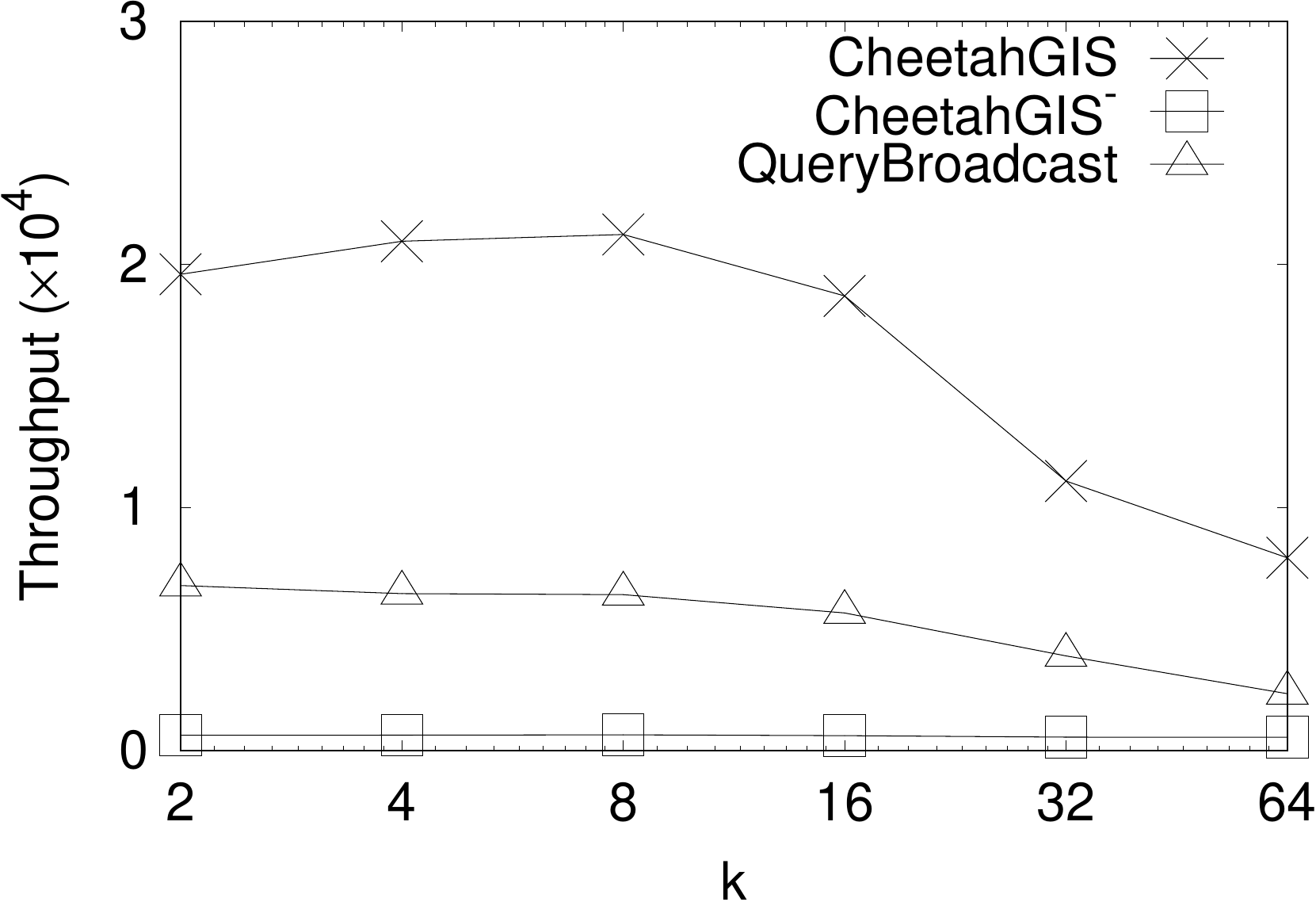}
    &
    \includegraphics[width=0.63\columnwidth]{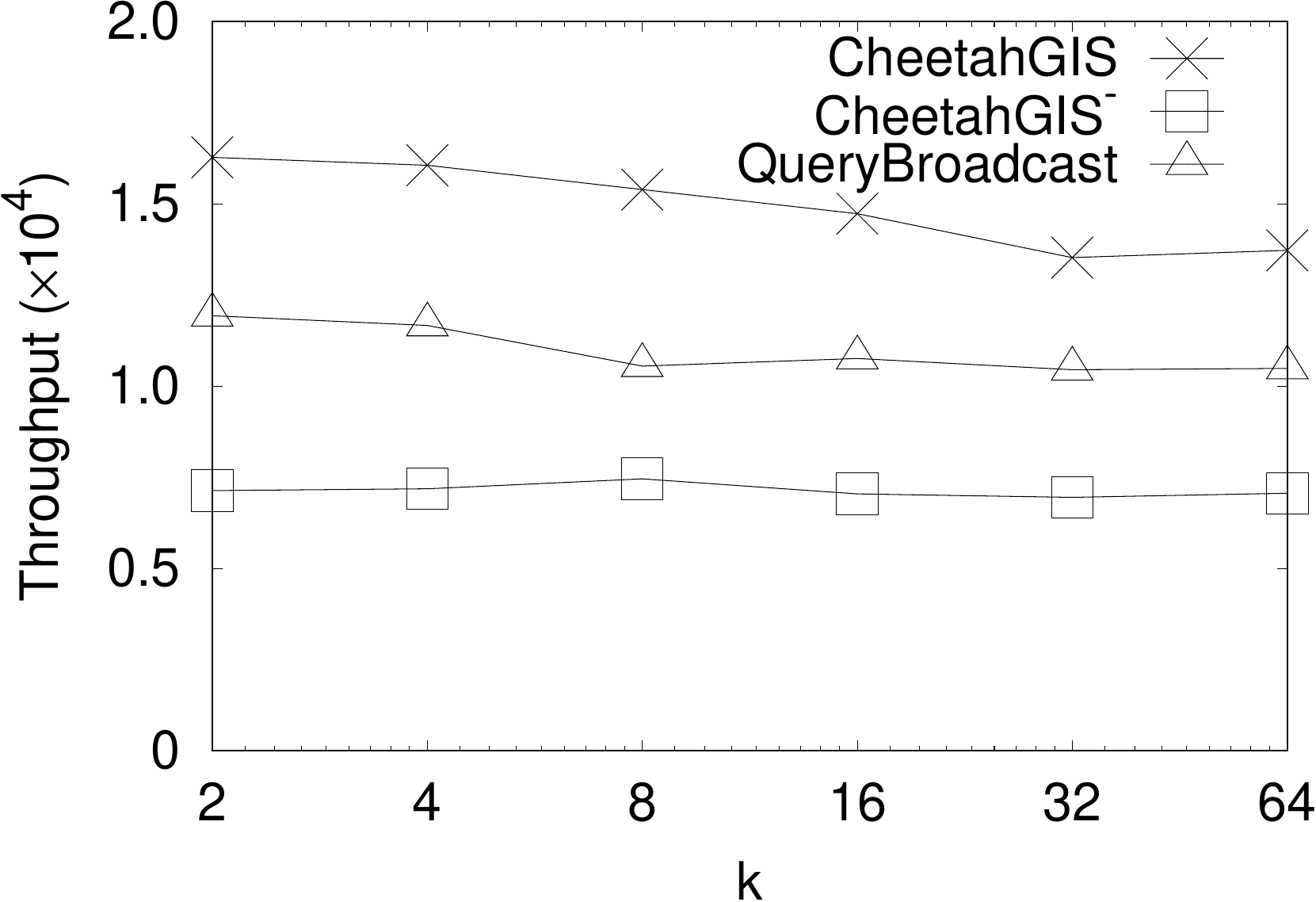}
    &
    \includegraphics[width=0.63\columnwidth]{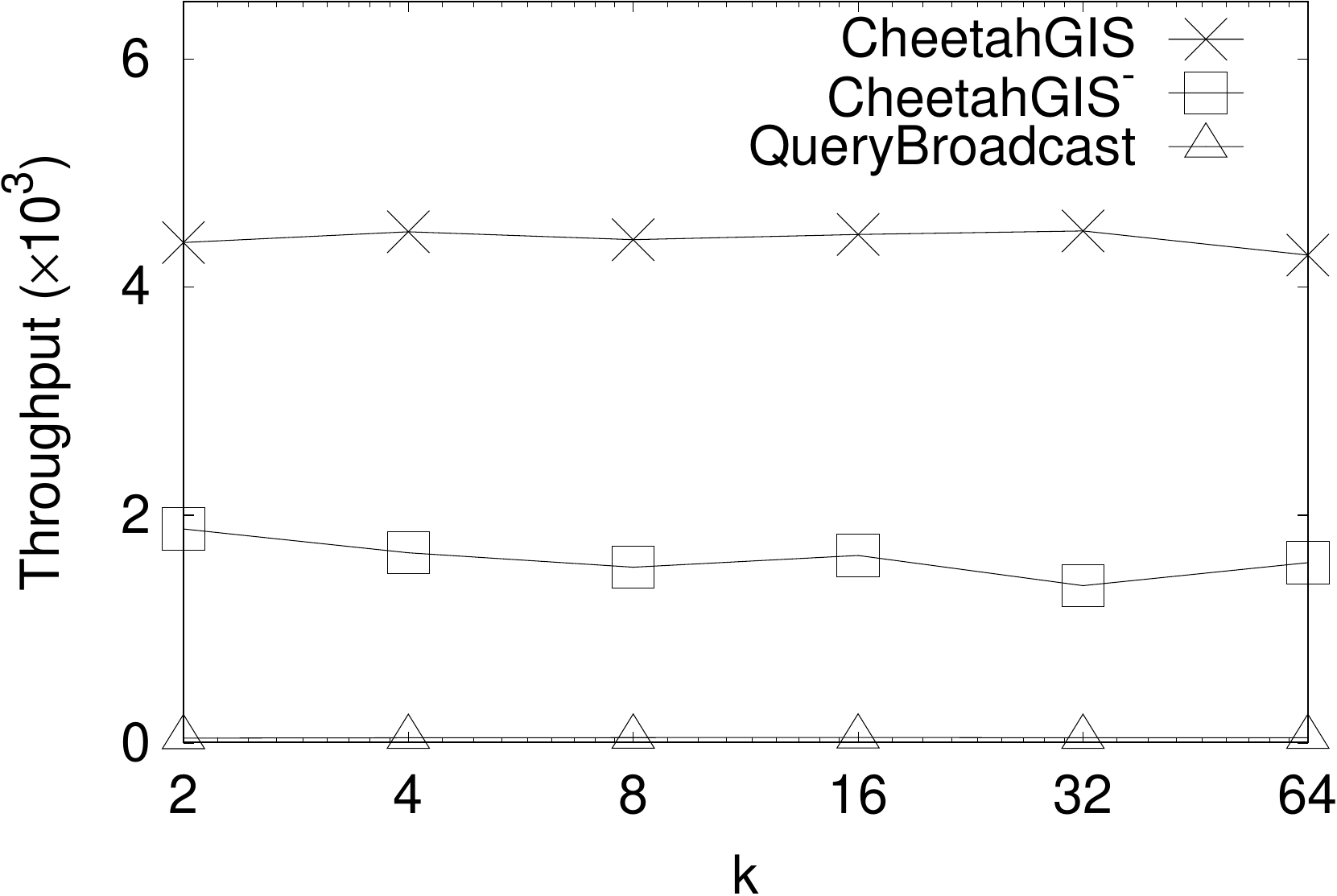}
    \\
    (a) \textsf{Shopping}
    &
    (b) \textsf{Geolife}
    &
    (c) \textsf{Brinkhoff}
    \end{tabular}
    \caption{$k$NN query performance evaluation}
    \label{fig:knneva}
\end{figure*}

\stitle{Compared solutions}

To evaluate the performance of \thor{}, we compare the following solutions.

\squishlist
    \item \bb: The general-purpose approach implementation based on \cite{zhang2016real}. It uses \statefun{}-based architecture but only broadcasts the query to every local processor and aggregates all the results from all \textsf{Local Processors} to obtain the result. 
    \item \thorm: It is a simple version of \thor. In particular, we evaluate the query processing performance of \thor{} by disabling its \textsf{Metadata Synchronizer} pruning techniques.
    \item \thor:It is our proposal in this work. All the discussed techniques and optimizations are enabled.
\squishend

We did not compare the performance with GeoFlink~\cite{shaikh2020geoflink}  as it processes queries fixed before runtime with window-based data streams, but our queries are only known during the runtime.


\stitle{Platform and measurements}
We use a distributed cluster with 10 working nodes (Intel Xeon Gold 5122 3.60GHz) in the experiments.
We use one of the working nodes as the master node, another one as Kafka broker (which produces and consumes data stream for \thor{}), and the rest are computing nodes.
We use Ubuntu 18.04 as the operating system and the version of Flink, Statefun, and Kafka are 1.11.2, 2.2.0, and 2.7.0, respectively.
For each computing node, we use 64GB memory and deploy 15 task slots.
All system modules and algorithms are implemented in Java.

We report two performance measurements (i.e., throughput and latency) in all experiments.
The throughput is defined as the total time elapsed (i.e., the total time from the first query received to the last query processed) divided by the total number of processed queries.
The latency is the average time cost of each query, i.e., from the query submitted to the system to its result returned.
All reported measurements are taking an average of three repeated experiments.

\subsection{Overall Performance Evaluation}\label{sec:overall}
We first evaluate the \thor{}'s performance of range count query processing on three different datasets by varying the query region size in Figure~\ref{fig:rngeva}.
First,  the throughput of \bb{} is stable by varying the region sizes as \bb{} broadcasts each query to every local processor.
Second, the throughput of both \thor{} and \thorm{} performs up to 1.2 and 1.9 times faster than \bb{} in all cases.
The reason is both \thor{} and \thorm{} prunes unqualified cells for each range count query via the grid-based index in \textsf{Indexer} module.
Third, the throughput of both \thor{} and \thorm{} falls with the increase of query region size in all three datasets.
The reason is that the larger the query region, the more qualified candidate cells.

Interestingly, the performance of range count query processing in \thor{} is much better than its in \thorm{} with \textsf{Geolife} and \textsf{Brinkhoff} datasets, as shown in Figures~\ref{fig:rngeva}(b) and (c), respectively.
However, both \thor{} and \thorm{} perform similar for range count query processing in \textsf{Shopping}, see Figure~\ref{fig:rngeva}(a).
The reason is that the pruning ability of the \textsf{Indexer} module is almost the same in \thor{} and \thorm{} with \textsf{Shopping} dataset,
but it is quite different in \thor{} and \thorm{} when the dataset is \textsf{Geolife} or \textsf{Brinkhoff}.
In particular, all the moving objects are randomly distributed in the ATC shopping center, and the number of qualified candidate cells does not change too much for a given query region in both \thor{} and \thorm.
However, the moving objects are distributed on the road network of Beijing and Las Vegas in \textsf{Geolife} and \textsf{Brinkhoff}, respectively.
Thus, the pruning ability of \thor{} is better than \thorm{} as \thor{} enables the \textsf{Metadata Synchronizer} module.

In Figure~\ref{fig:knneva}, we measure the \thor{}'s performance of $k$NN query processing. 
We use a large enough radius in \thorm{} to guarantee the correctness of $k$NN query processing as it cannot identify the minimum circle region via \textsf{Metadata Synchronizer} module.
As illustrated in Figure~\ref{fig:knneva}, \thor{} is always the winner in all datasets, which confirms the effectiveness of our proposed techniques (e.g., \textsf{Load Balancer}, \textsf{Indexer}, \textsf{Metadata Synchronizer}).
In addition, the $k$NN query performance of \thorm{} is  better than \bb{} on \textsf{Brinkhoff}, because \bb{} is very inefficient when it need scan the large amount of moving objects in \textsf{Brinkhoff}.
Last but not least, the throughput of $k$NN query in \thor{} and \thorm{} does not affect too much with the rising of $k$ in \textsf{Geolife} and \textsf{Brinkoff},
but it drops obviously in \textsf{Shopping} as the moving object distribution is more random in \textsf{Shopping} and the number of qualified candidates cells are increasing significantly when $k$ becomes larger in it.

\begin{figure*}
	\hspace*{-1cm}
    \centering
    \begin{tabular}{cccc}
    \includegraphics[width=0.5\columnwidth]{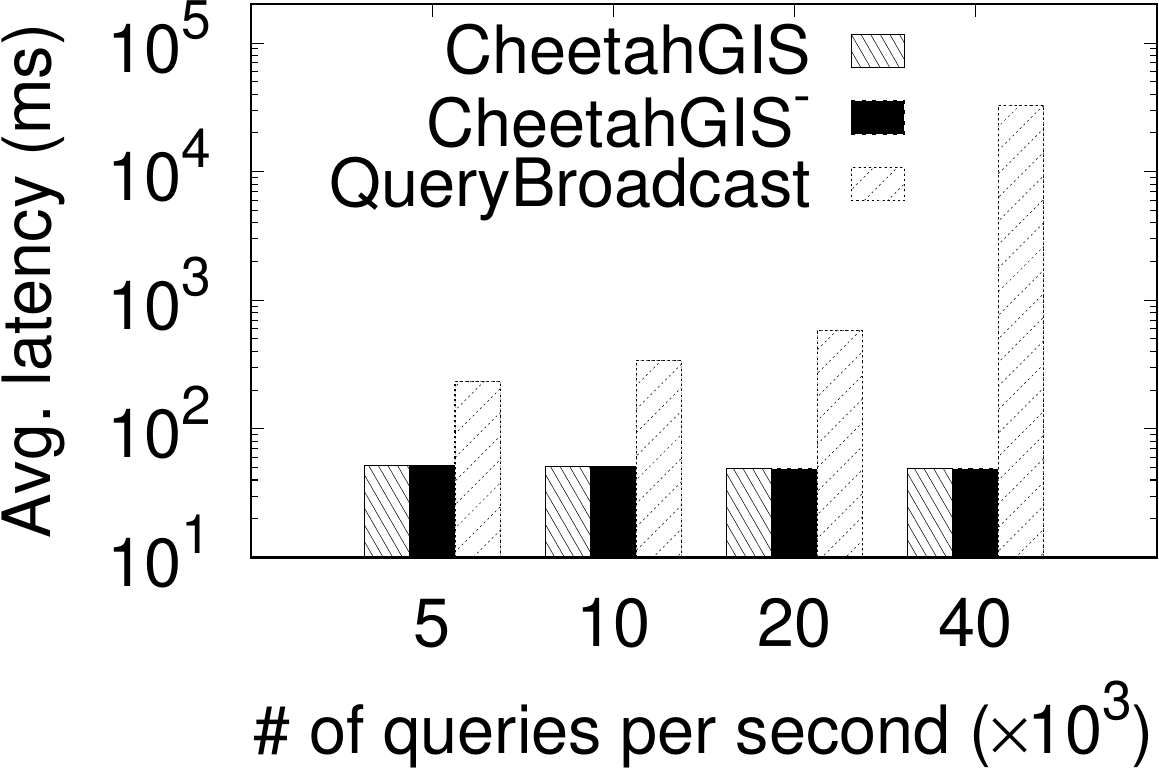}
    &
    \includegraphics[width=0.5\columnwidth]{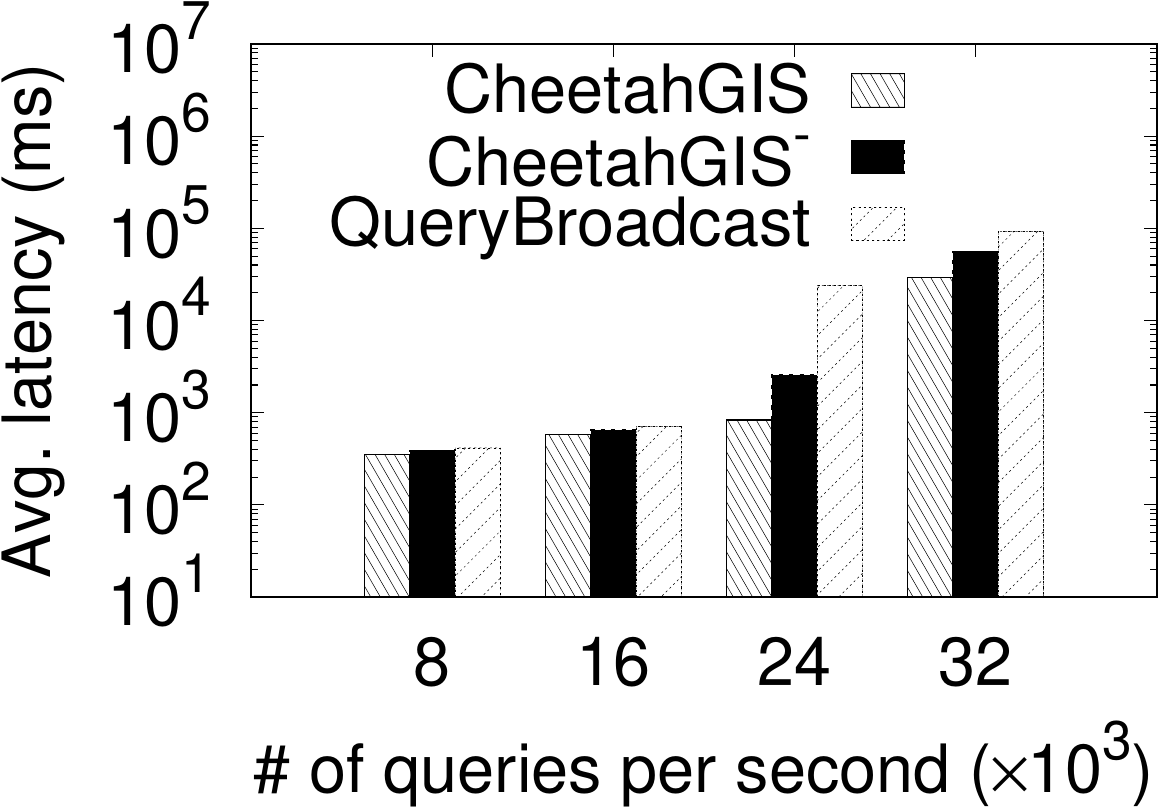}
    &
    \includegraphics[width=0.5\columnwidth]{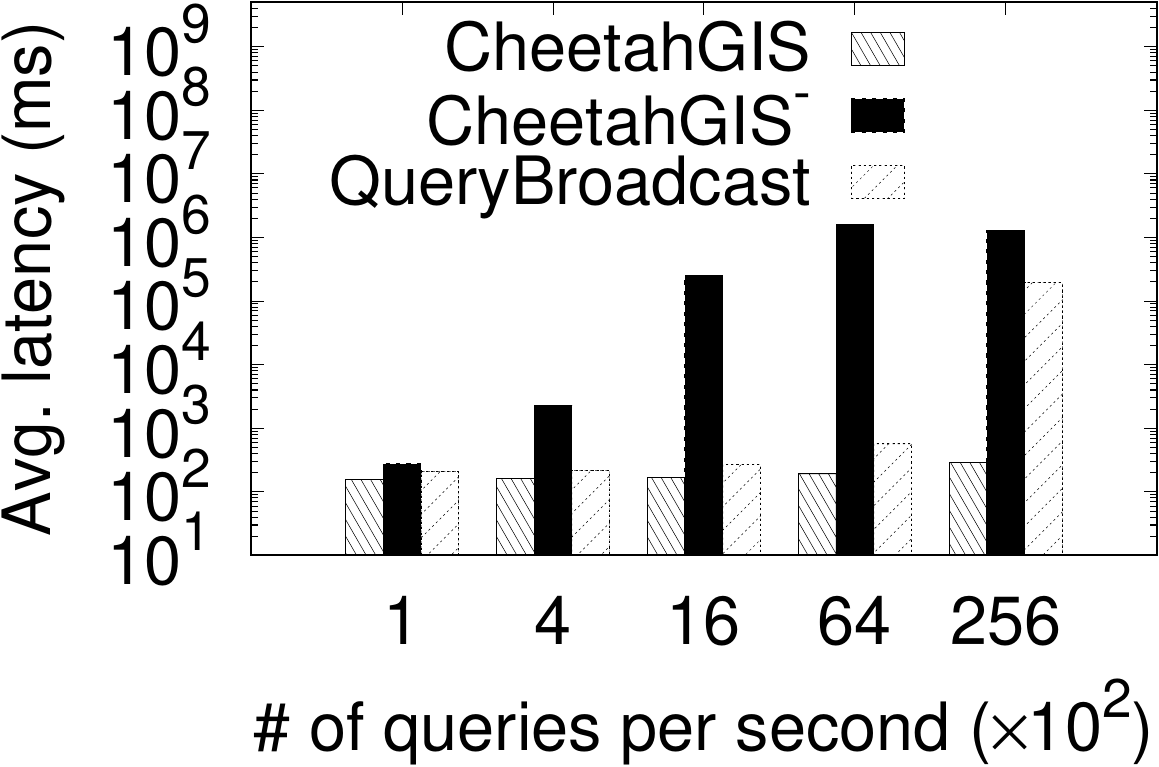}
    &
    \includegraphics[width=0.53\columnwidth]{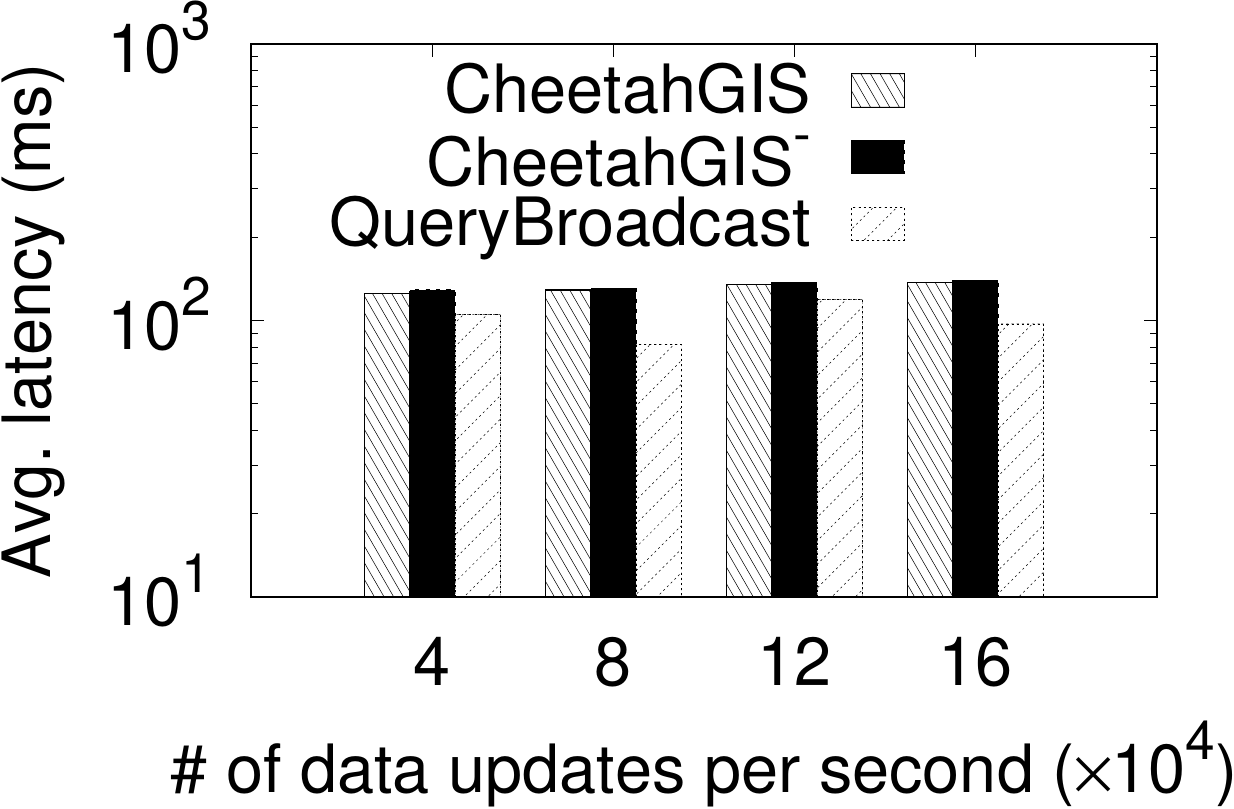}
    \\
    (a) Object query
    &
    (b) Range count query
    &
    (c) $k$NN query
    &
    (d) Data updates
    \end{tabular}
	\caption{Latency evaluation of \thor{} by varying throughput requirement, on \textsf{Shopping}}
	\label{fig:latency}
\end{figure*}

In Figure~\ref{fig:latency}, we report the average latency of the corresponding queries or operations by considering different numbers of queries/data updates per second issued within 5 minutes.
Specifically, we evaluate the average latency of object query, range count query, $k$NN query, and data object updates on \textsf{Shopping} dataset.
In Figure~\ref{fig:latency}(a), we report the average latency of object query, which contains randomly selected moving object ID in \textsf{Shopping}.
The latency of \bb{} is significantly larger than both \thor{} and \thorm, and it takes almost 10 seconds when the number of issued queries is 4000 per second.
Both \thor{} and \thorm{} perform excellently, i.e., the object query processing latency is only 50ms in all cases, as the object query can be answered directly by looking up the hash table in \textsf{Transformer} module.
In Figure~\ref{fig:latency}(b), we measure the average latency of range count query with query region ratio 0.16\% on \textsf{Shopping}.
First, when the number of range count queries is 8000 to 16000 per second, the latency of all three solutions (i.e., \thor, \thorm{} and \bb) are similar.
However, the latency of range count query processing in \thorm{} and \bb{} is unacceptable when the number of range count queries is 24000 per second,
e.g., it takes 23.78s to process a range count query in \bb.
It means that the processing capacity of \thorm{} and \bb{} is below 24000 range count query per second.
Last but not least, \thor{} performs the best in all cases.

We plot the latency of $k$NN query processing in Figure~\ref{fig:latency}(c) by setting $k=8$.
First, when the number of $k$NN queries is around 100 per second, the average latency of all three compared solutions (i.e., \thor, \thorm, and \bb) is almost the same.
However, the performance of \thorm{} degenerates seriously when the number of $k$NN queries is larger than 400 per second, it performs even worse than \bb.
The reason is \thorm{} estimates a large enough radius to guarantee the correctness of $k$NN query result, which results in poor pruning ability of \textsf{Indexer} and heavy workload of each
\textsf{Local Processor}.
In addition, the latency of \bb{} is 566ms when the number of $k$NN queries is 6400 per second.
However, the latency of \thor{} is 292ms even the number of $k$NN queries turns to 25600 per second.
In Figure~\ref{fig:latency}(d), we measure the average data update latency of \thor, \thorm{} and \bb.
As expected, the latency of \thor{} and \thorm{} to handle data updates is 20-40ms slower than \bb.
This is because \thor{} and \thorm{} have more modules and additional processing on data, e.g., data transformation on \textsf{Transformer} and grid-based routing on \textsf{Indexer},
and \bb{} is a very simple structure and only shuffles data to storage module.
Nevertheless, the data update latencies of all three solutions are quite good, i.e., it only takes almost 100ms to ingest all the data whose coming frequency is 160K data per second.

\subsection{Effect of Design Choices}\label{sec:effect}
In this section, we investigate the effect of different design choices (e.g., cell size in grid-based index, cluster size, execution mode, resource management, and load balancer) in \thor.

\begin{figure}
    \centering
    \begin{tabular}{cc}
    \includegraphics[width=0.48\columnwidth]{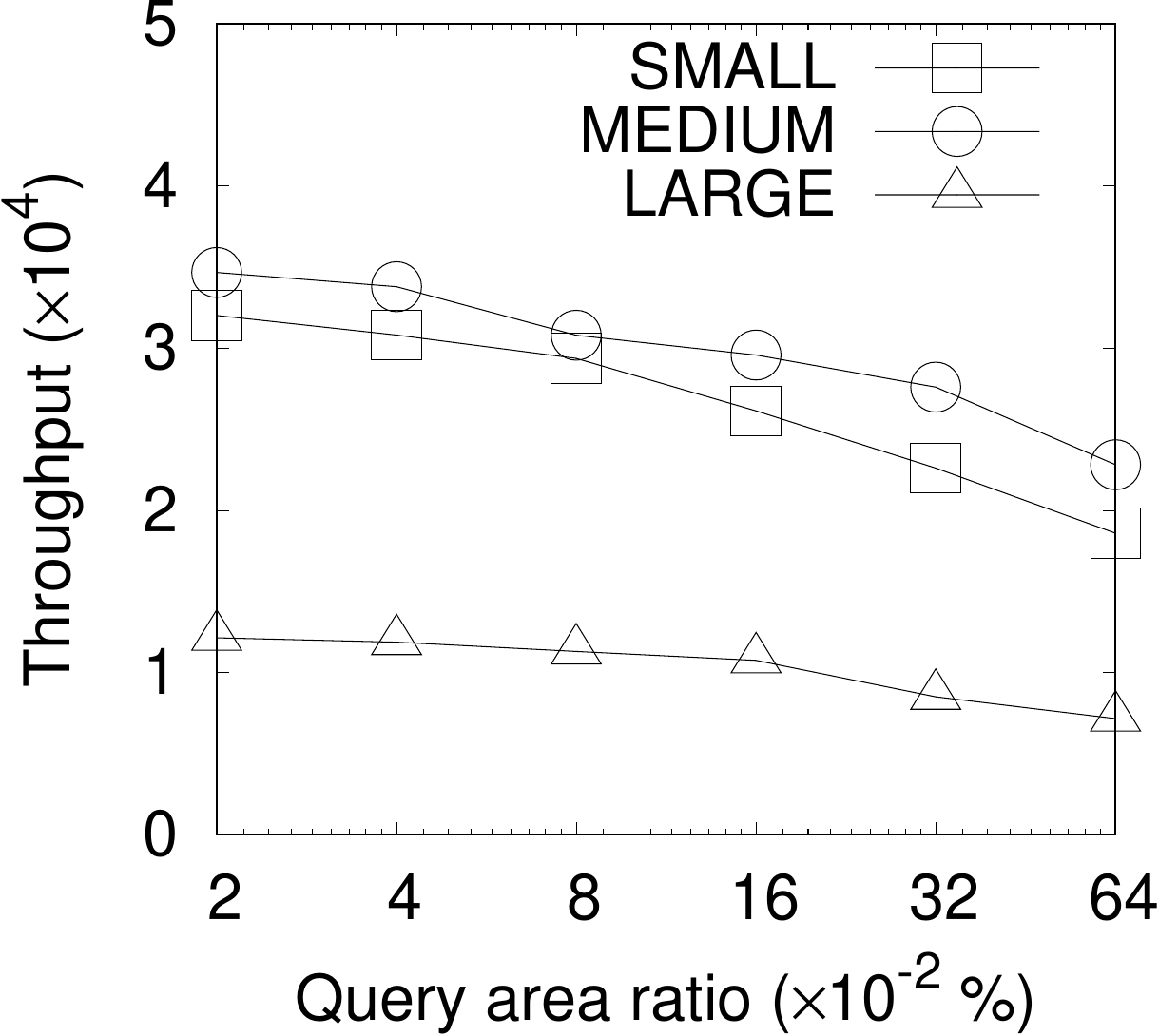}
    &
    \includegraphics[width=0.458\columnwidth]{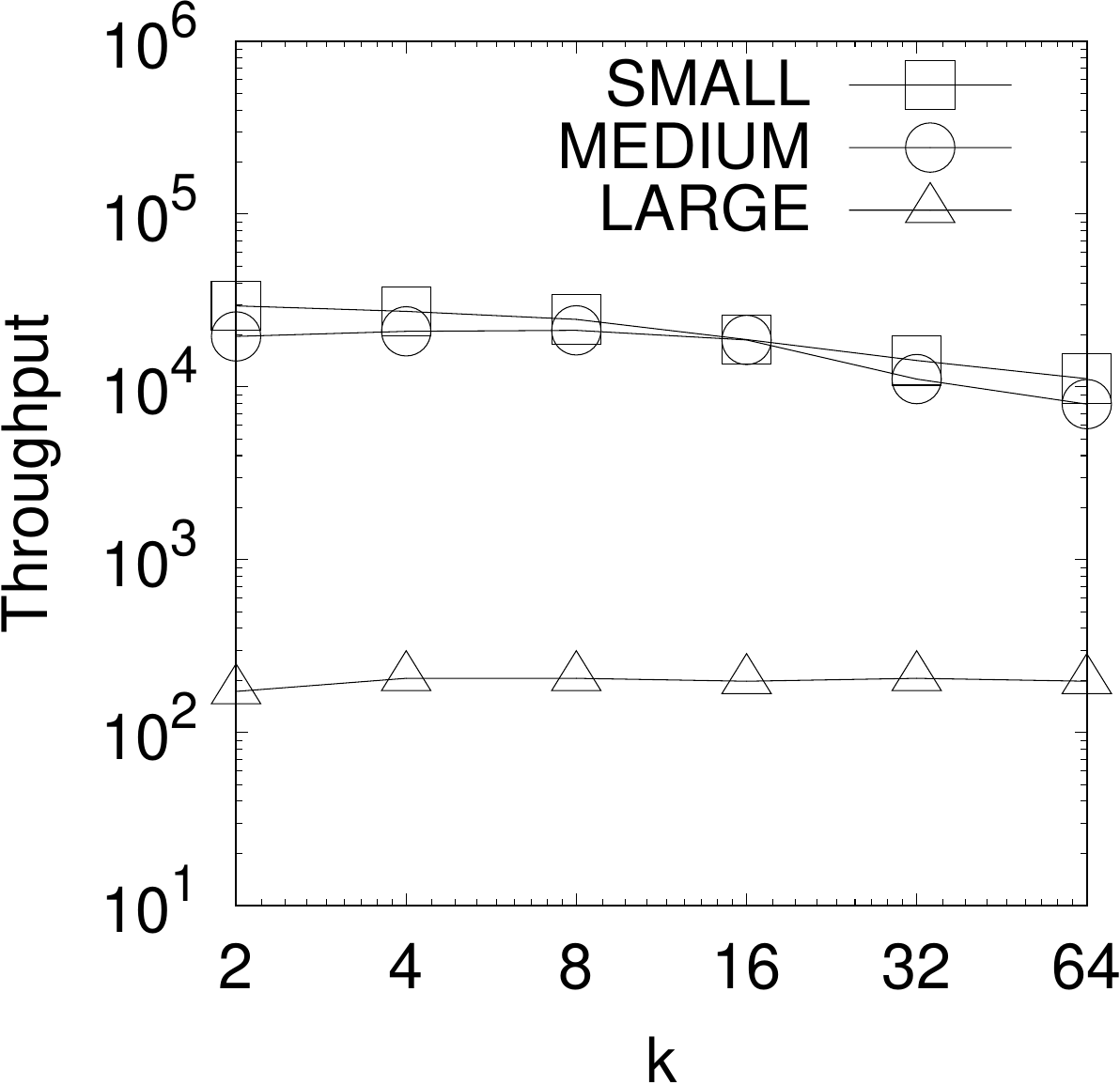}
    \\
    (a) Range count query
    &
    (b) $k$NN query
    \end{tabular}
	\caption{Throughput (queries per second) varying grid-cell size, on \textsf{Shopping}}
	\label{fig:cellsize}
\end{figure}

We first evaluate the throughput of range count query and $k$NN query with different cell sizes in the grid-based index in \thor (see Section~\ref{sec:indexer}).
In particular, we test three different cell sizes, i.e., \textsf{SMALL: $0.6m \times 0.4m$}, \textsf{DEFAULT: $1.0m \times 0.667m$}, and \textsf{LARGE: $2.4m \times 1.6m$}.
The corresponding experimental results are shown in Figures~\ref{fig:cellsize}(a) and (b), respectively.
For range count query, their throughputs fall with the rising of query region size in all three cell sizes, as shown in Figure~\ref{fig:cellsize}(a).
the largest cell size (i.e., \textsf{LARGE}) performs the worst in all cases as the pruning ability of the grid-based index with a large cell size is poor and many cells will be processed by \textsf{Local Processor} instances.
The throughput of the smallest cell size (i.e., \textsf{SMALL}) is slightly worse than the throughput of the default cell size
as the overhead of cell pruning becomes obvious when the cell size is too small, i.e., fine-grained pruning is expensive.

In Figure~\ref{fig:cellsize}(b), we show the throughput of $k$NN query with three different levels of cell size in the grid-based index.
The overall trend is the throughput falls when the value of $k$ increases.
In addition, similar to the performance of range count query, the largest cell size (i.e., \textsf{LARGE}) performs the worst in $k$NN query processing.
Interestingly, the smallest cell size (i.e., \textsf{SMALL}) has the best performance when $k$ ranges from 2 to 64 as the pruning ability of small cell size in the grid-based index is very excellent.
It reduces many computation costs as only fewer qualified candidate cells will be processed in \textsf{Local Processor} instance.

Combing the experimental results in Figures~\ref{fig:cellsize}(a) and (b), it is not trivial to set a proper cell size to support the efficient range count query and $k$NN query simultaneously in \thor.
We suggest the users use a medium cell size, e.g., as the default size in our experiments.

\begin{figure}
    \centering
    \begin{tabular}{cc}
    \includegraphics[width=0.48\columnwidth]{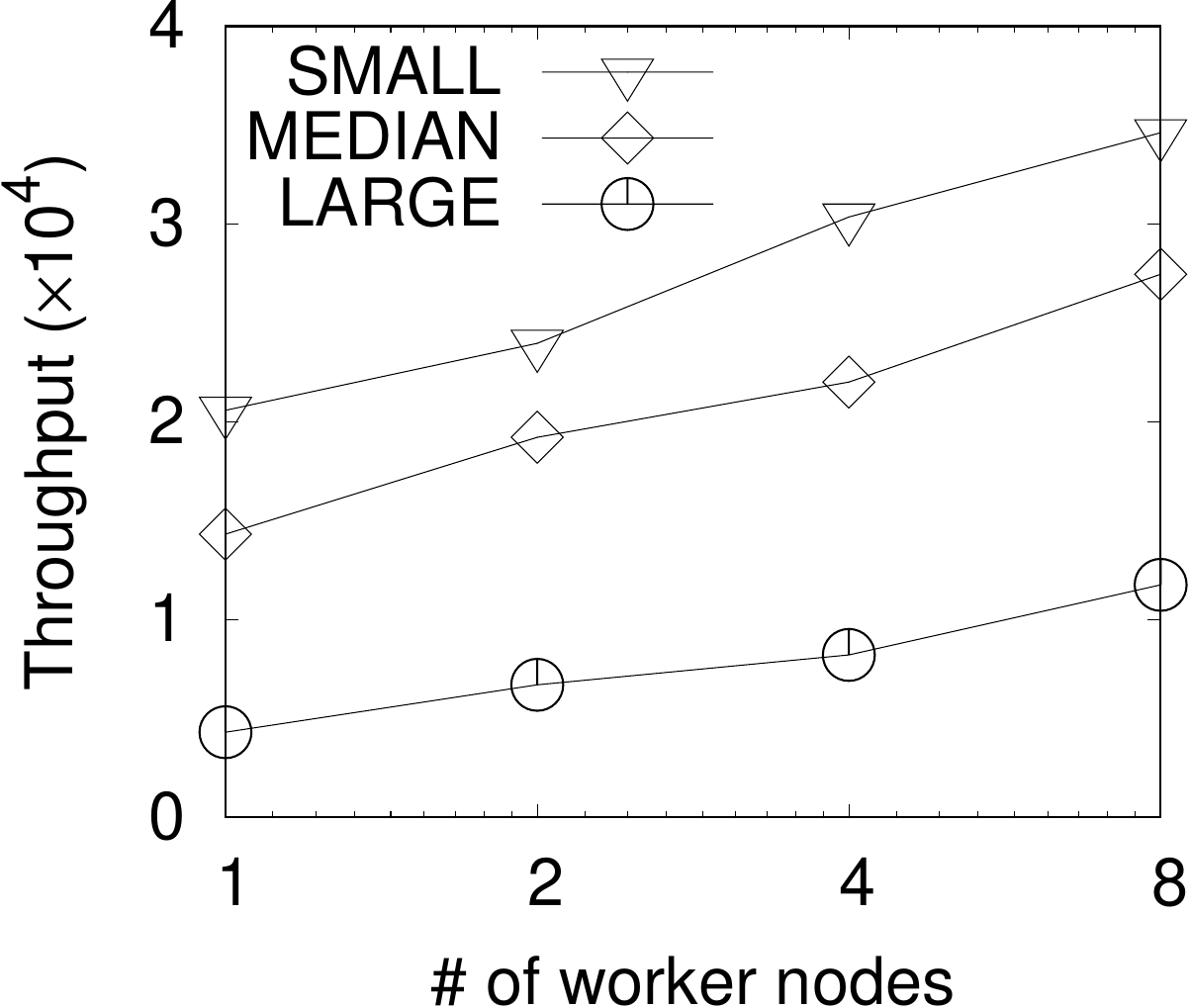}
    &
    \includegraphics[width=0.48\columnwidth]{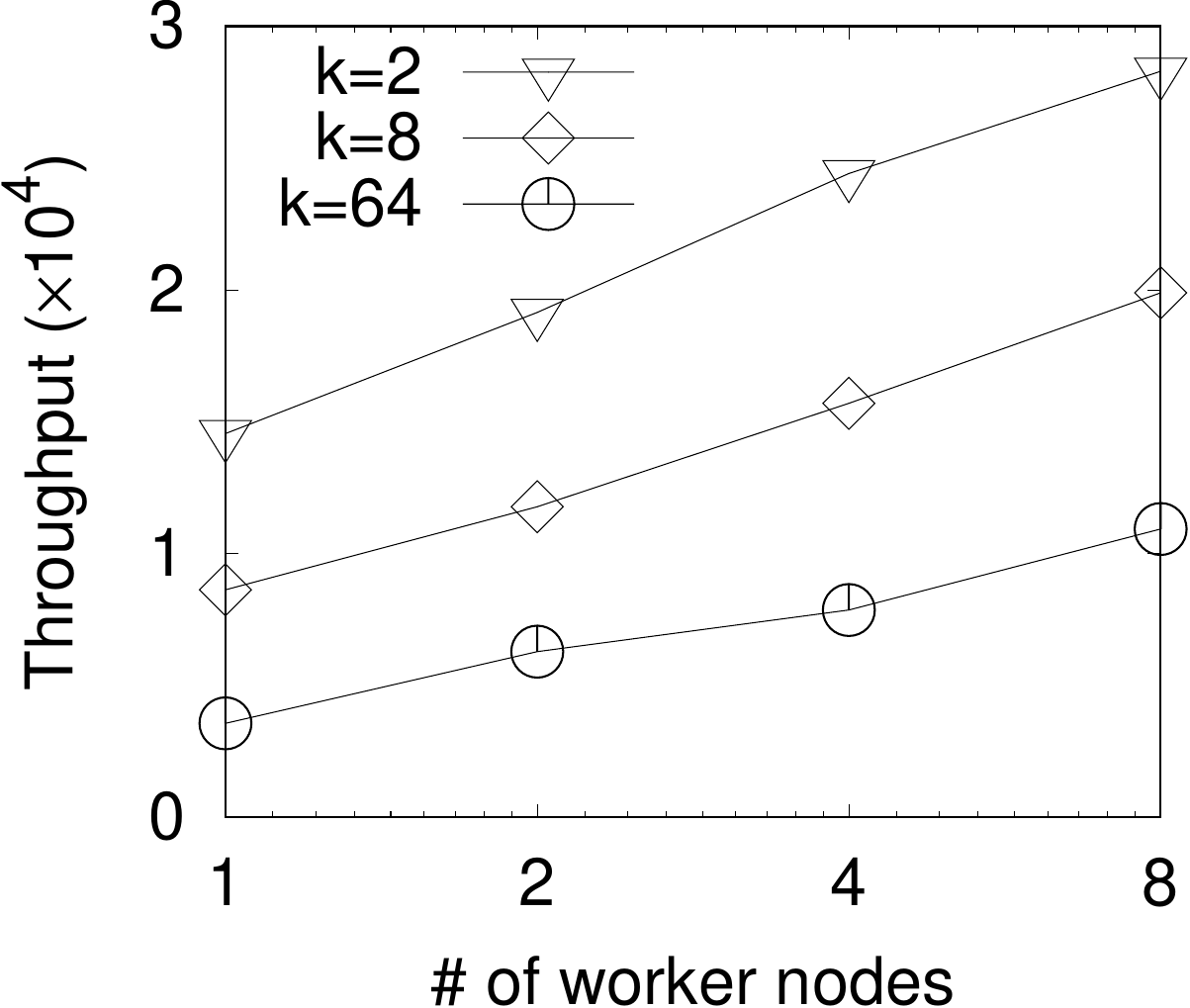}
    \\
    (a) Range count query
    &
    (b) $k$NN query
    \end{tabular}
	\caption{Throughput (queries per second) varying the number of workers, on \textsf{Shopping}}
	\label{fig:parallelism}
\end{figure}

We then investigate the effect of the number of workers in the cluster in Figure~\ref{fig:parallelism}.
In particular, the number of workers ranges from 2 to 8 to test the throughput of range count query and $k$NN query in \thor.
The number of task slots is 15 in each worker and the default data set is \textsf{Shopping}.
As illustrated in Figure~\ref{fig:parallelism}(a), the throughput of the range count query rises in all cases when the number of workers increases.
In addition, the throughput of range count query with large query region size (i.e., \textsf{LARGE: $64 \times 10^{-2}\%$}) is always lower than that with small query region size (i.e., \textsf{SMALL: $2 \times 10^{-2}\%$}).
The throughput of $k$NN query are plotted in Figure~\ref{fig:parallelism}(b).
Similarly, the throughput of $k$NN query in all cases is increasing with the rising in the number of workers.
The performance of $k$NN with large $k$ (i.e., $k=64$) is worse than that with small $k$ (i.e., k =2).

\begin{figure}
    \centering
    \begin{tabular}{cc}
    \includegraphics[width=0.48\columnwidth]{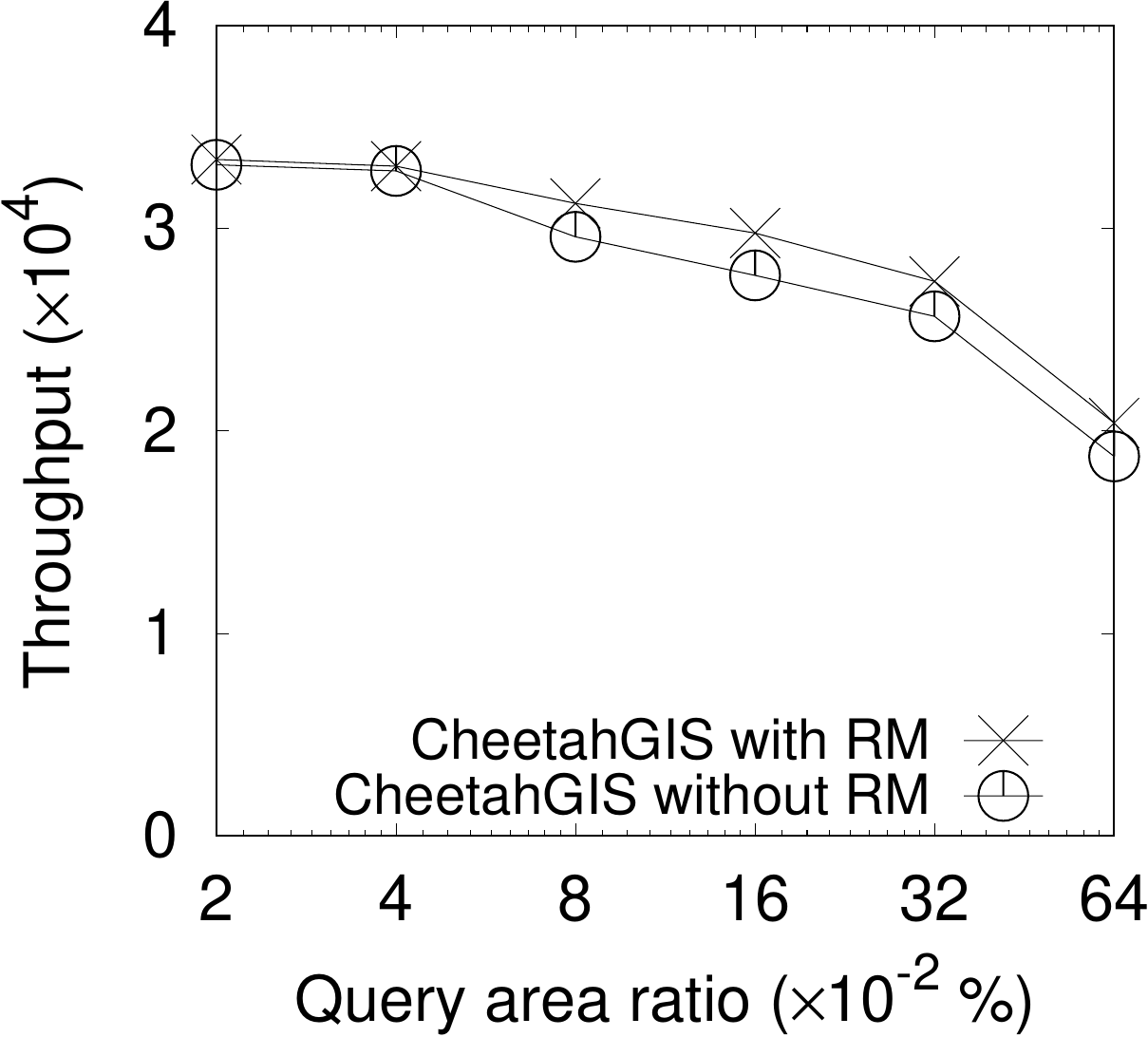}
    &
    \includegraphics[width=0.463\columnwidth]{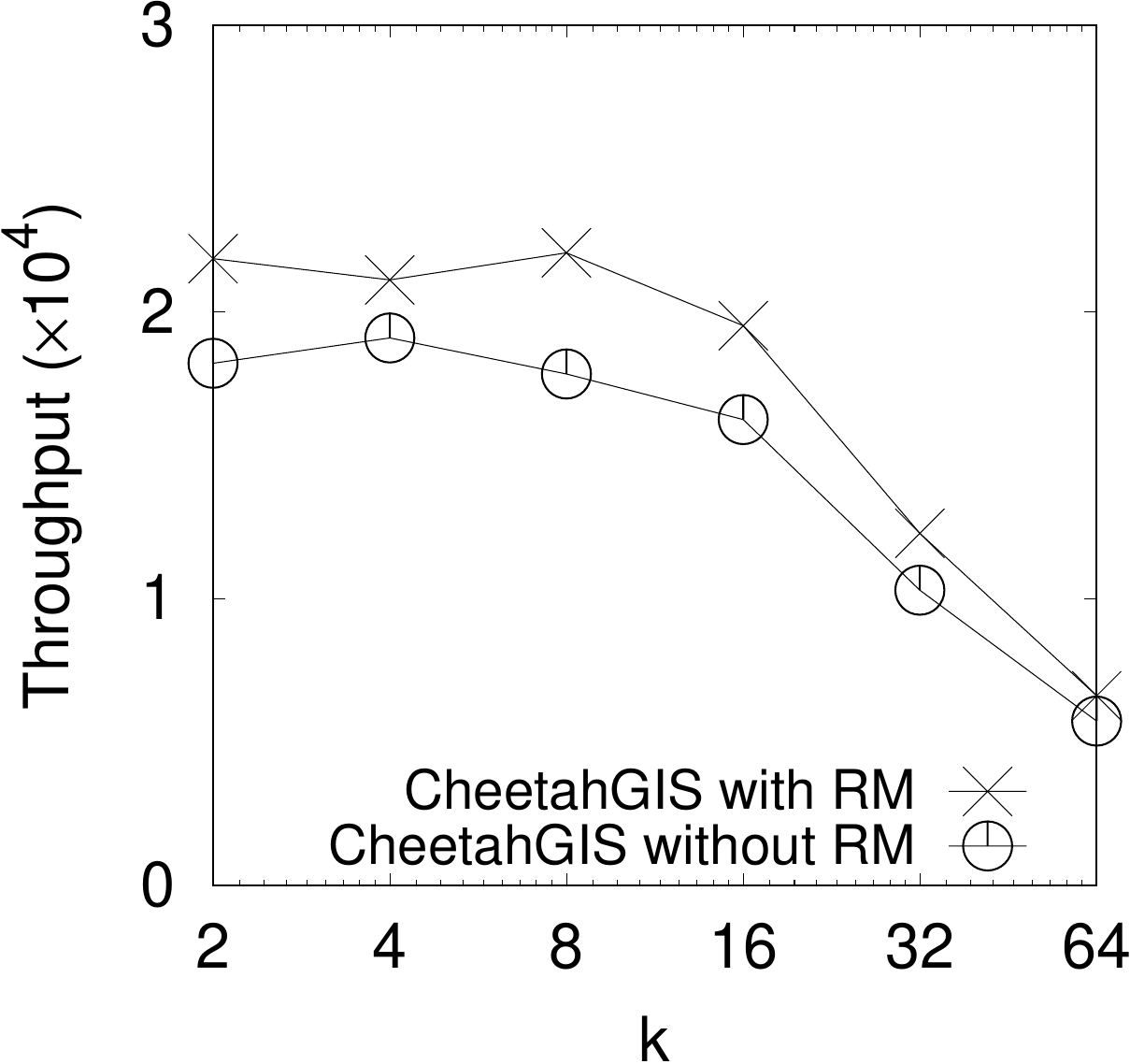}
    \\
    (a) range count query
    &
    (b) $k$NN query
    \end{tabular}
	\caption{Throughput (queries per second) varying resource management method, on \textsf{Shopping}}
	\label{fig:rm}
\end{figure}

We study the effect of fine granularity resource management in \thor{} (see Section~\ref{sec:rm}) with range count query and $k$NN query in Figures~\ref{fig:rm}(a) and (b), respectively.
Specifically, the \textsf{\thor{} with RM} and \textsf{\thor{} without RM} refers to the \thor{} with and without fine granularity resource management scheme.
The \textsf{\thor{} with RM} outperforms \textsf{\thor{} without RM} by 5\%-9\% and 10\%-23\% for range count query and $k$NN query, respectively.
The reason is obvious as our fine-granularity resource management avoids the resource contention in the original task slot allocation scheme in Flink, i.e., \textsf{\thor{} without RM}.

\begin{figure}
    \centering
    \begin{tabular}{cc}
    \includegraphics[width=0.48\columnwidth]{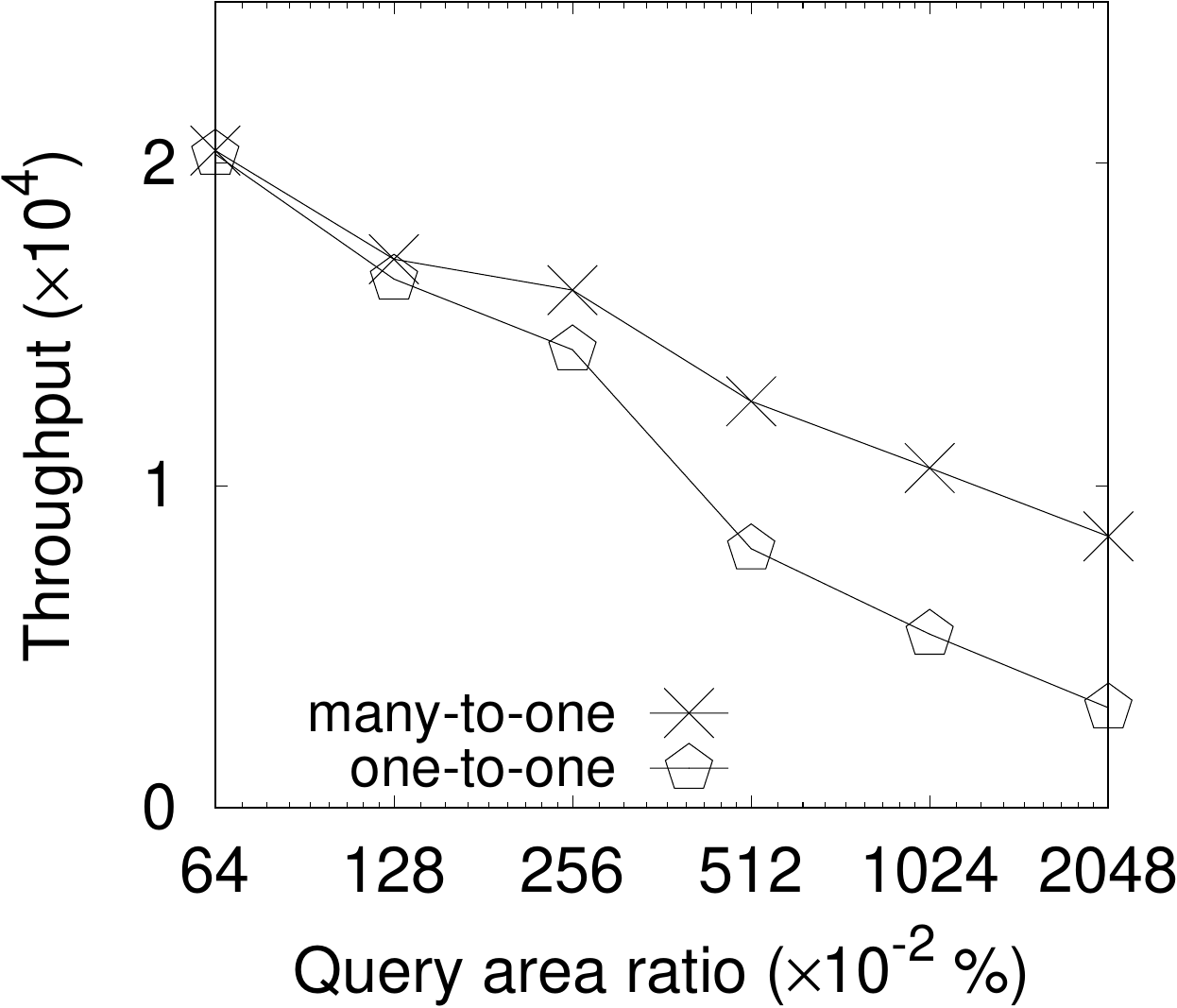}
    &
    \includegraphics[width=0.45\columnwidth]{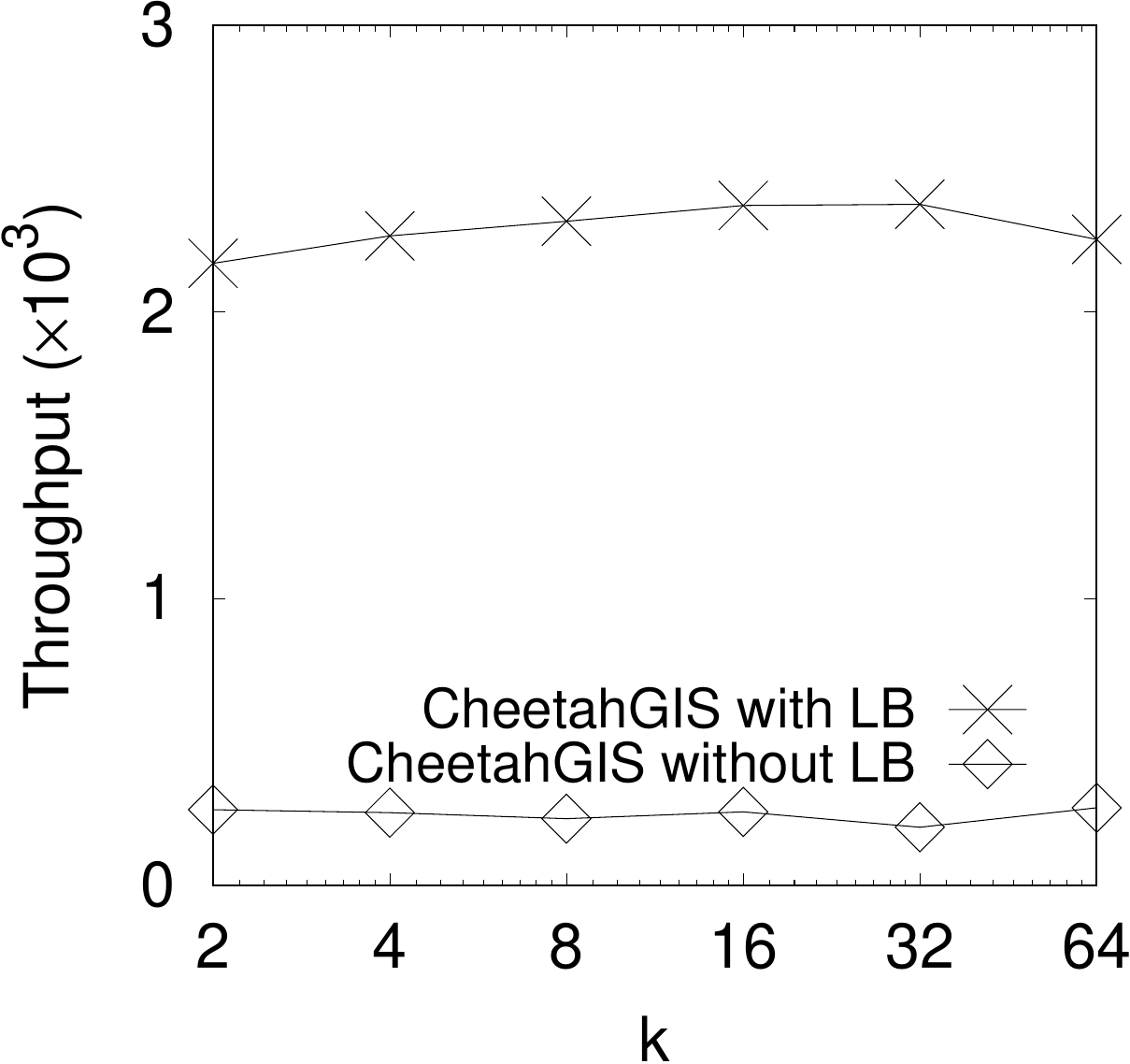}
    \\
    (a) range count query
    &
    (b) $k$NN query
    \end{tabular}
    \trim
	\caption{Throughput (queries per second) varying design of \thor{}}
	\label{fig:effect}
	\trim
\end{figure}

In Figure~\ref{fig:effect}(a), we measure the performance of range count query in \thor{} with \textsf{Shopping} dataset by considering two \textsf{Local Processor} execution modes: \textit{many-to-one} (i.e., the execution mode in \thor) and \textit{one-to-one} (i.e., the alternative mode we discussed in Section~\ref{sec:cellgroup}).
The performance of range count query degenerates in both execution modes when the query region size becomes larger.
However, \textit{many-to-one} execution mode is always better than that of \textit{one-to-one}, and the performance gap becomes obvious when the query region size is large.
The reason is that \textit{many-to-one} execution mode reduces the communication cost by processing different cells in a cell group, i.e., \textit{many-to-one} execution mode.

We last evaluate the effect of \textsf{Load Balancer} techniques (see Section~\ref{sec:loadbalance}) in Figure~\ref{fig:effect}(b).
In particular, we compared the performance of \thor{} with or without \textsf{Load Balancer} module when processing $k$NN query, as \textsf{\thor{} with LB} and \textsf{\thor{} without LB} in Figure~\ref{fig:effect}(b).
To exemplify, we generate a synthetic dataset with \textsf{Shopping} to enlarge the data bias.
Specifically, 95\% of moving objects in \textsf{Shopping} are concentrated on these cells which are managed by  3 out of 120 \textsf{Local Processors}.
As shown in Figure~\ref{fig:effect}(b), \textsf{\thor{} with LB} improves \textsf{\thor{} without LB} 8-10 times in terms of the throughput.
It confirms the superiority of our proposed load balancing techniques in Section~\ref{sec:loadbalance}.

\section{Conclusion}\label{sec:con}
In this work, we proposed \thor{} to process streaming spatial queries (e.g., range count query, $k$NN query, object query) efficiently.
It is built upon \statefun{} and streaming system Flink.
It consists of 6 modules, i.e., \textsf{Transformer}, \textsf{Indexer}, \textsf{Local Processor},
\textsf{Aggregator}, \textsf{Load Balancer} and \textsf{Metadata Synchronizer}.
The advantages of \thor{} can be summarized as follows:
(i) the excellent query processing performance on a wide range of spatial streaming queries;
(ii) its extensible system architecture and generic query processing procedure;
and (iii) the superiority of the designed techniques and optimizations (e.g., \textit{many-to-one} execution mode, fine-granularity resource management).
We demonstrate the effectiveness of spatial streaming query processing in \thor{} by extensive experiments on various datasets.
In the future, \thor{} will be enhanced from two aspects.
The first one is devising advanced techniques (efficient imbalance remedy algorithm) to improve the system performance,
and the second one is supporting streaming spatial analytical query processing (e.g., outlier region detection, spatial data clustering) effectively.

\bibliographystyle{IEEEtran}
\bibliography{ref}

\end{document}